\documentclass{article}

\usepackage{a4}
\usepackage{authblk}
\usepackage{amsmath,amssymb}
\usepackage{amsthm}
\usepackage{tabls}
\usepackage[dvipdfmx]{graphicx}
\usepackage{wrapfig}
\usepackage{array}
\usepackage{url}
\usepackage{multirow}
\usepackage{multicol}
\usepackage{here}
\usepackage{algpseudocode}

\usepackage[draft,mode=multiuser]{fixme}
\fxuseenvlayout{colorsig}
\fxusetargetlayout{color}
\FXRegisterAuthor{ss}{ass}{SS}
\FXRegisterAuthor{si}{asi}{SI}
\FXRegisterAuthor{hb}{ahb}{HB}
\FXRegisterAuthor{mt}{amt}{MT}
\fxsetface{env}{}

\newtheorem{theorem}{Theorem}
\newtheorem{lemma}{Lemma}
\newcommand{\RLE}{\mathit{RLE}}
\newcommand{\Diff}{\mathit{diff}}
\newcommand{\Beg}{\mathit{beg}}
\newcommand{\End}{\mathit{end}}

\newcommand{\RLEBoundary}{\mathit{RLE\_Bound}}
\newcommand{\Succ}{\mathit{succ}}
\newcommand{\bp}{\mathrm{bp}}

\newcommand{\Parikh}{\mathcal{P}}

\title{
  Computing Abelian regularities on RLE strings
}

\author[1]{Shiho Sugimoto}
\author[2]{Naoki Noda}
\author[1]{Shunsuke Inenaga}
\author[1]{Hideo Bannai}
\author[1]{Masayuki Takeda}

\affil[1]{
  Department of Informatics, Kyushu University, Japan\\
   \{shiho.sugimoto, inenaga, bannai, takeda\}@inf.kyushu-u.ac.jp
  }
\affil[2]{
  Department of Physics, Kyushu University, Japan
}

\begin{document}
\maketitle

\begin{abstract}
  Two strings $x$ and $y$ are said to be Abelian equivalent
  if $x$ is a permutation of $y$, or vice versa.
  If a string $z$ satisfies $z = xy$ with $x$ and $y$
  being Abelian equivalent,
  then $z$ is said to be an \emph{Abelian square}.
  If a string $w$ can be factorized into a sequence $v_1, \ldots, v_s$
  of strings
  such that $v_1$, \ldots, $v_{s-1}$ are all Abelian equivalent
  and $v_s$ is a substring of a permutation of $v_1$,
  then $w$ is said to have a \emph{regular Abelian period} $(p, t)$
  where $p = |v_1|$ and $t = |v_s|$.
  If a substring $w_1[i..i+\ell-1]$ of a string $w_1$ and
  a substring $w_2[j..j+\ell-1]$ of another string $w_2$ are Abelian equivalent,
  then the substrings are said to be a common Abelian factor of $w_1$ and $w_2$
  and if the length $\ell$ is the maximum of such then the substrings are
  said to be a \emph{longest common Abelian factor} of $w_1$ and $w_2$.
  We propose efficient algorithms which compute these Abelian regularities 
  using the \emph{run length encoding (RLE)} of strings.
  For a given string $w$ of length $n$ whose RLE is of size $m$,
  we propose algorithms which compute all Abelian squares occurring in $w$
  in $O(mn)$ time,
  and all regular Abelian periods of $w$ in $O(mn)$ time.
  For two given strings $w_1$ and $w_2$ of total length $n$
  and of total RLE size $m$,
  we propose an algorithm which computes
  all longest common Abelian factors in $O(m^2n)$ time.
\end{abstract}

\section{Introduction}

Two strings $s_1$ and $s_2$ are said to be \emph{Abelian equivalent}
if $s_1$ is a permutation of $s_2$, or vice versa.
For instance, strings $ababaac$ and $caaabba$ are Abelian equivalent.
Since the seminal paper by Erd{\H o}s~\cite{Erdos61} published in 1961,
the study of Abelian equivalence on strings has attracted much attention,
both in word combinatorics and string algorithmics.

\subsection{Our problems and previous results}

In this paper, we are interested in the following algorithmic problems
related to Abelian regularities of strings.
\begin{enumerate}
 \item \label{prob:1} Compute \emph{Abelian squares} in a given string.
 \item \label{prob:2} Compute \emph{regular Abelian periods} of a given string.
 \item \label{prob:3} Compute \emph{longest common Abelian factors} of two given strings.
\end{enumerate}

Cummings and Smyth~\cite{Cummings_weakrepetitions}
proposed an $O(n^2)$-time algorithm to solve Problem~\ref{prob:1},
where $n$ is the length of the given string.
Crochemore et al.~\cite{crochemore13:_abelian} proposed
an alternative $O(n^2)$-time solution to the same problem.
Recently, Kociumaka et al.~\cite{KociumakaRW15} showed
how to compute all Abelian squares in $O(s + \frac{n^2}{\log^2 n})$ time,
where $s$ is the number of outputs.

Related to Problem~\ref{prob:2},
various kinds of Abelian periods of strings have been considered:
An integer $p$ is said to be a \emph{full Abelian period} of a string $w$
iff there is a decomposition $u_1, \ldots, u_z$ of $w$
such that $|u_i| = p$ for all $1 \leq i \leq z$
and $u_1, \ldots, u_z$ are all Abelian equivalent.
A pair $(p, t)$ of integers is said to be a \emph{regular Abelian period}
(or simply an \emph{Abelian period}) of a string $w$
iff there is a decomposition $v_1, \ldots, v_s$ of $w$
such that $p$ is a full Abelian period of $v_1 \cdots v_{s-1}$,
$|v_i| = p$ for all $1 \leq i \leq s-1$,
and $v_s$ is a permutation of a substring of $v_1$ (and hence $t \leq p$).
A triple $(h, p, t)$ of integers is said to be a \emph{weak Abelian period} of
a string $w$ iff there is a decomposition $y_1, \ldots, y_r$ of $w$
such that $(p, t)$ is an Abelian period of $y_2 \cdots y_{r}$,
$|y_1| = h$, $|y_i| = p$ for all $2 \leq i \leq r-1$, $|y_r| = t$,
and $y_1$ is a permutation of a substring of $y_2$ (and hence $h \leq p$).

The study on Abelian periodicity of strings was initiated by
Constantinescu and Ilie~\cite{ConstantinescuI06}.
Fici et al.~\cite{FiciLLPS16} gave
an $O(n \log \log n)$-time algorithm
to compute all full Abelian periods.
Later, Kociumaka et al.~\cite{DBLP:conf/stacs/KociumakaRR13} showed
an optimal $O(n)$-time algorithm to compute all full Abelian periods.

Fici et al.~\cite{FiciLLPS16} also showed
an $O(n^2)$-time algorithm to compute all regular Abelian periods
for a given string of length $n$.
Kociumaka et al.~\cite{DBLP:conf/stacs/KociumakaRR13}
also developed an algorithm which finds all regular Abelian periods
in $O(n(\log\log n + \log \sigma))$ time,
where $\sigma$ is the alphabet size.

Fici et al.~\cite{FiciLLP14} proposed an algorithm
which computes all weak Abelian periods in $O(\sigma n^2)$ time,
and later Crochemore et al.~\cite{crochemore13:_abelian}
proposed an improved $O(n^2)$-time algorithm to compute
all weak Abelian periods.
Kociumaka et al.~\cite{KociumakaRW15} showed how to compute
all \emph{shortest} weak Abelian periods in $O(n^2/\sqrt{\log n})$ time.

Consider two strings $w_1$ and $w_2$.
A pair $(s_1, s_2)$ of a substring $s_1$ of $w_1$
and a substring $s_2$ of $w_2$ is said to be a \emph{common Abelian factor}
of $w_1$ and $w_2$,
iff $s_1$ and $s_2$ are Abelian equivalent.
Alatabbi et al.~\cite{alatabbi:_algor_longes_common_abelian_factor}
proposed an $O(\sigma n^2)$-time and $O(\sigma n)$-space algorithm
to solve Problem~\ref{prob:3} of computing
all \emph{longest common Abelian factors} (\emph{LCAFs}) of 
two given strings of total length $n$.
Later, Grabowski~\cite{Grabowski15} showed an algorithm
which finds all LCAFs in $O(\sigma n^2)$ time with $O(n)$ space.
He also presented an $O((\frac{\sigma}{k} + \log \sigma) n^2 \log n)$-time
$O(kn)$-space algorithm for a parameter $k \leq \frac{\sigma}{\log \sigma}$.
Recently, Badkobeh et al.~\cite{BadkobehGGNPS16} proposed
an $O(n \log^2 n \log^* n)$-time $O(n \log^2 n)$-space algorithm
for finding all LCAFs.

\subsection{Our contribution}

In this paper, we show that we can accelerate computation of
Abelian regularities of strings via \emph{run length encoding} (\emph{RLE})
of strings.
Namely, if $m$ is the size of the RLE of a given string $w$ of length $n$,
we show that:
\begin{enumerate}
  \item[(1)] All Abelian squares in $w$ can be computed in $O(mn)$ time.
  \item[(2)] All regular Abelian periods of $w$ can be computed in $O(mn)$ time.
\end{enumerate}
Since $m \leq n$ always holds,
solution (1) is at least as efficient as 
the $O(n^2)$-time solutions by Cummings and Smyth~\cite{Cummings_weakrepetitions}
and by Crochemore et al.~\cite{crochemore13:_abelian},
and can be much faster when the input string $w$ is highly compressible by RLE.

Amir et al.~\cite{AmirAHLLR14} proposed an $O(\sigma(m^2+n))$-time
algorithm to compute all Abelian squares using RLEs.
Our $O(mn)$-time solution is faster than theirs when
$\frac{\sigma m^2}{m-\sigma} = \omega(n)$.

Solution (2) is faster than the $O(n(\log\log n + \log \sigma))$-time
solution by Kociumaka et al.~\cite{DBLP:conf/stacs/KociumakaRR13}
for highly RLE-compressible strings with $\log \log n = \omega(m)$\footnote{Since we can w.l.o.g. assume that $\sigma \leq m$, the $\log \sigma$ term is negligible here.}.

Also, if $m$ is the total size of the RLEs of two given strings $w_1$ and $w_2$
of total length $n$, we show that:
\begin{enumerate}
  \item[(3)] All longest common Abelian factors of $w_1$ and $w_2$
    can be computed in $O(m^2n)$ time.
\end{enumerate}
Our solution (3) is faster than the $O(\sigma n^2)$-time solution
by Grabowski~\cite{Grabowski15} when $\sigma n = \omega(m^2)$,
and is faster than the fastest variant of the other solution by Grabowski~\cite{Grabowski15} (choosing $k = \frac{\sigma}{\log \sigma}$) when $\sqrt{n \log n \log \sigma} = \omega(m)$.
Also, solution (3) is faster than the $O(n \log^2 n \log^* n)$-time
solution by Badkobeh et al.~\cite{BadkobehGGNPS16} when $\log n \sqrt{\log^* n} = \omega(m)$.
The time bounds of our algorithms are all deterministic.

\section{Preliminaries}

Let $\Sigma = \{c_1, \ldots, c_\sigma\}$ be an ordered alphabet
of size $\sigma$.
An element of $\Sigma^*$ is called a string.
For any string $w$, $|w|$ denotes the length of $w$.
The empty string is denoted by $\varepsilon$.
Let $\Sigma^{+} = \Sigma^{*}-\{\varepsilon\}$.
For any $1 \leq i \leq |w|$,
$w[i]$ denotes the $i$-th symbol of $w$.
For a string $w=xyz$, strings $x$, $y$, and $z$ are called 
a \emph{prefix}, \emph{substring}, and \emph{suffix} of $w$, respectively.
The substring of $w$ that begins at position $i$ and ends at position $j$ 
is denoted by $w[i..j]$ for $1 \leq i \leq j \leq |w|$.
For convenience, let $w[i..j] = \varepsilon$ for $i > j$.

For any string $w\in\Sigma^*$, 
its \emph{Parikh vector} $\Parikh_w$ is an array of length $\sigma$
such that for any $1 \leq i \leq |\Sigma|$,
$\Parikh_w[i]$ is the number of occurrences of each character $c_i\in \Sigma$ in $w$.
For example, for string $w = abaab$ 
over alphabet $\Sigma = \{a, b\}$,
$\Parikh_w = \langle 3,2 \rangle$.
We say that strings $x$ and $y$ are \emph{Abelian equivalent}
if $\Parikh_x = \Parikh_y$.
Note that $\Parikh_x = \Parikh_y$
iff $x$ and $y$ are permutations of each other.
When $x$ is a substring of a permutation of $y$, 
then we write $\Parikh_x \subseteq \Parikh_y$.
For any Parikh vectors $P$ and $Q$,
let $\Diff(P, Q) = |\{i \mid P[i] \neq Q[i], 1 \leq i \leq \sigma\}|$.

A string $w$ of length $2k > 0$ is called an \emph{Abelian square}
if it is a concatenation of two Abelian equivalent strings of length $k$ each,
i.e., $\Parikh_{w[1..k]} = \Parikh_{w[k+1..2k]}$.
A string $w$ is said to have a \emph{regular Abelian period} $(p, t)$
if $w$ can be factorized into a sequence $v_1, \ldots, v_s$ of substrings
such that $p = |v_1| = \cdots = |v_{s-1}|$, $|v_s| = t$,
$\Parikh_{v_i} = \Parikh_{v_1}$ for all $2 \leq i < s$,
and $\Parikh_{v_s} \subseteq \Parikh_{v_1}$.

For any strings $w_1, w_2 \in \Sigma^*$,
if a substring $w_1[i..i+l-1]$ of $w_1$ and
a substring $w_2[j..j+l-1]$ of $w_2$ are Abelian equivalent,
then the pair of substrings is said to be
a \emph{common Abelian factor} of $w_1$ and $w_2$.
When the length $l$ is the maximum of such then
the pair of substrings is said to be a
\emph {longest common Abelian factor} of $w_1$ and $w_2$.

The \emph{run length encoding} (\emph{RLE}) of string $w$ of length $n$,
denoted $\RLE(w)$, is 
a compact representation of $w$ which encodes 
each maximal character run $w[i..i+p-1]$ by $a^p$, 
if (1) $w[j] = a$ for all $i \leq j \leq i+p-1$,
(2) $w[i-1] \neq w[i]$ or $i = 1$, and (3) $w[i+p-1] \neq w[i+p]$ or $i+p-1 = n$.
E.g., $\RLE(aabbbbcccaaa\$) = a^2b^4c^3a^3\$^1$.
The {\em size} of $\RLE(w) = a_1^{p_1} \cdots a_m^{p_m}$ is the number $m$ of maximal character runs in $w$,
and each $a_i^{p_i}$ is called an \emph{RLE factor} of $\RLE(w)$.
Notice that $m \leq n$ always holds.
Also, since at most $m$ distinct characters can appear in $w$,
in what follows we will assume that $\sigma \leq m$.
Even if the underlying alphabet is large,
we can sort the characters appearing in $w$ in $O(m \log m)$ time
and use this ordering in Parikh vectors.
Since all of our algorithms will require at least $O(mn)$ time,
this $O(m \log m)$-time preprocessing is negligible.

For any $1 \leq i \leq j \leq n$,
let $\RLE(w)[i..j] = a^{p_i}_{i} \cdots a^{p_j}_{j}$.
For convenience let $\RLE(w)[i..j] = \varepsilon$ for $i > j$.
For $\RLE(w) = a_1^{p_1} \cdots a_m^{p_m}$,
let $\RLEBoundary(w) = \{1 + \sum_{i=1}^{k}p_k \mid 1 \leq k < m\} \cup \{1, n\}$.
For any $1 \leq i \leq n$,
let $\Succ(i) = \min\{j \in \RLEBoundary(w) \mid j > i \}$.
Namely, $\Succ(i)$ is the smallest position in $w$ 
that is greater than $i$
and is either the beginning position of an RLE factor in $w$
or the last position $n$ in $w$.

\section{Computing regular Abelian periods using RLEs}
\label{sec:Abelian_periods}

In this section, we propose an algorithm which computes
all regular Abelian periods of a given string.

\begin{theorem} \label{theo:Abelian_periods}
  Given a string $w$ of length $n$ over an alphabet of size $\sigma$,
  we can compute all regular Abelian periods of $w$ in $O(mn)$ time
  and $O(n)$ working space, where $m$ is the size of $\RLE(w)$.
\end{theorem}

\begin{proof}
Our algorithm is very simple.
We use a single window
for each length $d = 1, \ldots, \lfloor \frac{n}{2} \rfloor$.
For an arbitrarily fixed $d$,
consider a decomposition $v_1, \ldots, v_s$ of $w$
such that $v_i = w[(i-1)d+1..id]$ for
$1 \leq i \leq \lfloor \frac{n}{d} \rfloor$
and $v_s = w[n-(n \bmod d)+1..n]$.
Each $v_i$ is called a block,
and each block of length $d$ is called a complete block.

There are two cases to consider.

\noindent \textbf{Case (a)}: If $w$ is a unary string (i.e., $\RLE(w) = a^n$ for some $a \in \Sigma$).
In this case, 
$(d, (n \bmod d))$ is a regular Abelian period of $w$ for any $d$.
Also, note that this is the only case where 
$(d, (n \bmod d))$ can be a regular Abelian period of any string
of length $n$ with $\RLE(v_i) = a^d$ for some complete block $v_i$.
Clearly, it takes a total of $O(n)$ time and $O(1)$ space in this case.  

\noindent \textbf{Case (b)}: If $w$ contains at least two distinct characters,
then observe that a complete block $v_i$ is 
fully contained in a single RLE factor
iff $\Succ(1+(i-1)d) = \Succ(id)$.
Let $S$ be an array of length $n$ such that 
$S[j] = \Succ(j)$ for each $1 \leq j \leq n$.
We precompute this array $S$ in $O(n)$ time by a simple 
left-to-right scan over $w$.
Using the precomputed array $S$,
we can check in $O(m)$ time if there exists a complete block $v_i$ 
satisfying $\Succ(1+(i-1)d) = \Succ(id)$;
we process each complete block $v_i$ in increasing order of $i$
(from left to right), and stop 
as soon as we find the first complete block $v_i$
with $\Succ(1+(i-1)d) = \Succ(id)$.
If there exists such a complete block,
then we can immediately determine that $(d, (n \bmod d))$
is not a regular Abelian period (recall also Case (a) above.)

Assume every complete block $v_i$ overlaps at least two RLE factors.
For each $v_i$, let $m_i \geq 2$ be the number of
RLE factors of $\RLE(w)$ that $v_i$ overlaps
(i.e., $m_i$ is the size of $\RLE(v_i)$).
We can compute $\Parikh_{v_i}$ in $O(m_i)$ time from $\RLE(v_i)$,
using the exponents of the elements of $\RLE(v_i)$.
We can compare $\Parikh_{v_i}$ and $\Parikh_{v_{i-1}}$ in $O(m_i)$ time,
since there can be at most $m_i$ distinct
characters in $v_i$ and hence it is enough to check the $m_i$ entries
of the Parikh vectors.
Since there are $\lfloor \frac{n}{d} \rfloor$ complete blocks
and each complete block overlaps more than one RLE factor,
we have $\lfloor \frac{n}{d} \rfloor \leq \frac{1}{2}\sum_{i=1}^{s-1}m_i$.
Moreover, 
since each RLE factor is counted by a unique $m_i$
or by a unique pair of $m_{i-1}$ and $m_{i}$ for some $i$,
we have $\sum_{i=1}^sm_i \leq 2m$.
Overall, it takes $O(\sigma + \frac{n}{d} + \sum_{i=1}^{s}m_i) = O(m)$
time to test if
$(d, (n \bmod d))$ is a regular Abelian period of $w$.
Consequently, it takes $O(mn)$ total time to compute
all regular Abelian periods of $w$ for all $d$'s in this case.
Since we use the array $S$ of length $n$
and we maintain two Parikh vectors of the two adjacent $v_{i-1}$ and $v_{i}$
for each $i$,
the space requirement is $O(\sigma + n) = O(n)$.
\end{proof}

\begin{figure}[tb]
\centerline{
\includegraphics[width=6cm]{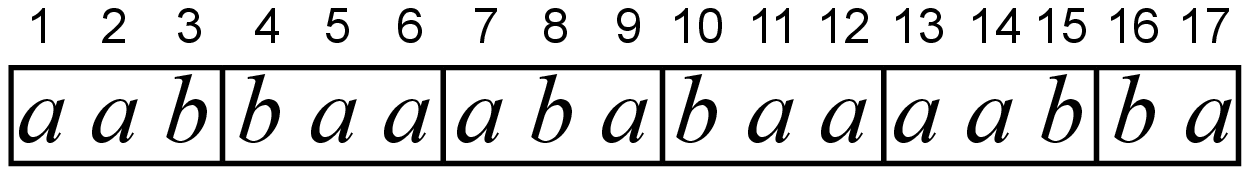}
}
\caption{$(3, 2)$ is a regular Abelian period of string
  $w = aabbaaababaaaabbaa$ 
  since $\Parikh_{w[1..3]} = \Parikh_{w[4..6]} = \Parikh_{w[7..9]} = \Parikh_{w[10..12]} = \Parikh_{w[13..15]} \supset \Parikh_{w[16..17]}$.}
\label{fig:period}
\end{figure}

For example, let $w = aabbaaababaaaabbaa$ and $d=3$.
See also Figure~\ref{fig:period} for illustration.
We have $RLE(w) = a^2b^2a^3b^1a^1b^1a^4b^2a^1$.
Then, we compute
$\Parikh_{v_1} = \langle 2,1 \rangle$ from $\RLE(v_1) = a^2b^1$,
$\Parikh_{v_2} = \langle 2, 1 \rangle$ from $\RLE(v_2) = b^1a^2$,
$\Parikh_{v_3} = \langle 2, 1 \rangle$ from $\RLE(v_3) = a^1b^1a^1$,
$\Parikh_{v_4} = \langle 2, 1 \rangle$ from $\RLE(v_4) = b^1a^2$,
$\Parikh_{v_5} = \langle 2, 1 \rangle$ from $\RLE(v_5) = a^2b^1$, and
$\Parikh_{v_6} = \langle 1, 1 \rangle$ from $\RLE(v_6) = b^1a^1$.
Since $\Parikh_{v_i} = \Parikh_{v_1}$ for $1 \leq i \leq 5$
and $\Parikh_{v_6} \subset \Parikh_{v_1}$,
$(3, 2)$ is a regular Abelian period of the string $w$.

\section{Computing Abelian squares using RLEs}
\label{sec:squares}

In this section, we describe our algorithm to compute 
all Abelian squares occurring in a given string $w$ of length $n$.
Our algorithm is based on the algorithm of 
Cummings and Smyth~\cite{Cummings_weakrepetitions}
which computes all Abelian squares in $w$ in $O(n^2)$ time.
We will improve the running time to $O(mn)$,
where $m$ is the size of $\RLE(w)$.

\subsection{Cummings and Smyth's $O(n^2)$-time algorithm}

We recall the
$O(n^2)$-time algorithm proposed by Cummings and Smyth~\cite{Cummings_weakrepetitions}.
To compute Abelian squares in a given string $w$, 
their algorithm aligns two adjacent sliding windows of length $d$
each, for every $1 \leq d \leq \lfloor \frac{n}{2} \rfloor$.

Consider an arbitrary fixed $d$.
For each position $1 \leq i \leq n-2d+1$ in $w$,
let $L_i$ and $R_i$ denote the left and right windows
aligned at position $i$.
Namely, $L_i = w[i..i+d-1]$ and $R_i = w[i+d..i+2d-1]$.
At the beginning, the algorithm computes 
$\Parikh_{L_1}$ and $\Parikh_{R_1}$ for position $1$ in $w$.
It takes $O(d)$ time to compute these Parikh vectors
and $O(\sigma)$ time to compute $\Diff(\Parikh_{L_1}, \Parikh_{R_1})$.
Assume $\Parikh_{L_i}$, $\Parikh_{R_i}$, and
$\Diff(\Parikh_{L_i}, \Parikh_{R_i})$ have been computed
for position $i \geq 1$,
and $\Parikh_{L_{i+1}}$, $\Parikh_{R_{i+1}}$, and
$\Diff(\Parikh_{L_{i+1}}, \Parikh_{R_{i+1}})$ is to be computed
for the next position $i+1$.
A key observation is that 
given $\Parikh_{L_i}$, then $\Parikh_{L_{i+1}}$
for the left window $L_{i+1}$ for the next position $i+1$
can be easily computed in $O(1)$ time,
since at most two entries of the Parikh vector can change.
The same applies to $\Parikh_{R_i}$ and $\Parikh_{R_{i+1}}$.
Also, given $\Diff(\Parikh_{L_i}, \Parikh_{R_i})$
for the two adjacent windows $L_i$ and $R_i$ for position $i$,
then it takes $O(1)$ time to determine whether or not 
$\Diff(\Parikh_{L_{i+1}}, \Parikh_{R_{i+1}}) = 0$
for the two adjacent windows $L_{i+1}$ and $R_{i+1}$ for the next position $i+1$.
Hence, for each $d$,
it takes $O(n)$ time to find all Abelian squares of length $2d$,
and thus it takes a total of $O(n^2)$ time for 
all $1 \leq d \leq \lfloor \frac{n}{2} \rfloor$.

\subsection{Our $O(mn)$-time algorithm}
We propose an algorithm which computes all Abelian squares in a given string $w$
of length $n$ in $O(mn)$ time, where $m$ is the size of $\RLE(w)$.

Our algorithm will output consecutive Abelian squares 
$w[i..i+2d-1]$, $w[i+1..i+2d]$, \ldots, $w[j..j+2d-1]$ of length $2d$ each
as a triple $\langle i, j, d \rangle$.
A single Abelian square $w[i..i+2d-1]$ of length $2d$ will be represented 
by $\langle i, i, d \rangle$.

For any position $i$ in $w$,
let $\Beg(L_i)$ and $\End(L_i)$ respectively denote
the beginning and ending positions of the left window $L_i$,
and let $\Beg(R_i)$ and $\End(R_i)$ respectively denote 
the beginning and ending positions of the right window $R_i$.
Namely, $\Beg(L_i) = i$, $\End(L_i) = i+d-1$,
$\Beg(R_i) = i+d$, and $\End(R_i) = i+2d-1$.
Cummings and Smyth's algorithm described above
increases each of $\Beg(L_i)$, $\End(L_i)$, $\Beg(R_i)$, and $\End(R_i)$
one by one, and tests all positions $i = 1, \ldots, n-2d+1$ in $w$. 
Hence their algorithm takes $O(n)$ time for each window size $d$.

In what follows,
we show that it is indeed enough to check only $O(m)$ positions in $w$
for each window size $d$.
The outline of our algorithm is as follows.
As Cummings and Smyth's algorithm,
we use two adjacent windows of size $d$,
and slide the windows.
However, unlike Cummings and Smyth's algorithm where the windows
are shifted by one position,
in our algorithm the windows can be shifted
by more than one position.
The positions that are not skipped and are explicitly examined
will be characterized by the RLE of $w$,
and the equivalence of the Parikh vectors of the two adjacent windows
for the skipped positions can easily be checked by simple arithmetics.

Now we describe our algorithm in detail.
First, we compute $\RLE(w)$ and let $m$ be its size.
Consider an arbitrarily fixed window length $d \geq 1$.

Initially, we compute $\Parikh_{L_1}$ and $\Parikh_{R_1}$ for position $1$.
We can compute these Parikh vectors in $O(m)$ time
and $O(\sigma)$ space using the same method
as in the algorithm of Theorem~\ref{theo:Abelian_periods}
in Section~\ref{sec:Abelian_periods}.

Then, we describe the steps for positions larger than $1$.
For each position $i \geq 1$ in a given string $w$,
let $D_1^i = \Succ(\Beg(L_i)) - \Beg(L_i)$, $D_2^i = \Succ(\Beg(R_i)) - \Beg(R_i)$,
and $D_3^i = \Succ(\End(R_i)+1) - \End(R_i)-1$.
The \emph{break point} for each position $i$, denoted $\bp(i)$,
is defined by $i + \min\{D_1^i, D_2^i, D_3^i\}$.
Assume the left window is aligned at position $i$ in $w$.
Then, we jump to the break point $\bp(i)$ directly from $i$.
In other words, the two windows $L_{i}$ and $R_{i}$
are directly shifted to $L_{\bp(i)}$ and $R_{\bp(i)}$, respectively.

It depends on the value of
$\Diff(\Parikh_{L_i}, \Parikh_{R_i})$ whether there can be
an Abelian square between positions $i$ and $\bp(i)$.
Note that $\Diff(\Parikh_{L_i}, \Parikh_{R_i}) \neq 1$.
Below, we characterize the other cases in detail.

\begin{lemma} \label{lem:diff=0}
  Assume $\Diff(\Parikh_{L_i}, \Parikh_{R_i}) = 0$.
  Then, for any $i < j \leq \bp(i)$,
  $j$ is the beginning position of an Abelian square of length $2d$
  iff $w[\Beg(L_i)] = w[\Beg(R_i)] = w[\End(R_i)+1]$.
\end{lemma}

\begin{proof}\label{pro:diff=0}
($\Leftarrow$)
By the definition of $\bp(i)$,
$w[\Beg(L_i)] = w[\Beg(L_j)]$,
$w[\Beg(R_i)] = w[\Beg(R_j)]$, and 
$w[\End(R_i)+1] = w[\End(R_j)+1]$ for all $i < j \leq \bp(i)$.
Let $c = w[\Beg(L_i)] = w[\Beg(R_i)] = w[\End(R_i)+1]$.
Then we have $w[\Beg(L_j)] = w[\Beg(R_j)] = w[\End(R_j)+1] = c$.
Thus the Parikh vectors of the sliding windows do not change
at any position between $i$ and $\bp(i)$. 
Since we have assumed $\Parikh_{L_i} = \Parikh_{R_i}$,
$\Parikh_{L_j} = \Parikh_{R_j}$ for any $i < j \leq \bp(i)$.
Thus $w[j..j+2d-1] = L_jR_j$ is an Abelian square of length $2d$
for any $i < j \leq \bp(i)$.

($\Rightarrow$)
Since $j$ is the beginning position of an Abelian square of length $2d$,
$\Parikh_{L_j} = \Parikh_{R_j}$.
Let $c_p = w[\Beg(L_i)]$, $c_q = w[\Beg(R_i)]$, and $c_t = w[\End(R_i)+1]$.
By the definition of $\bp(i)$,
$w[\Beg(L_j)] = c_p$,
$w[\Beg(R_j)] = c_q$, and 
$w[\End(R_j)+1] = c_t$ for any $i < j \leq \bp(i)$.
Also, for any $i < j \leq \bp(i)$,
$\Parikh_{L_j}[x] = \Parikh_{L_i}[x] - j + i$,
$\Parikh_{L_j}[y] = \Parikh_{L_i}[y] + j - i$,
$\Parikh_{R_j}[y] = \Parikh_{R_i}[y] - j + i$, and
$\Parikh_{R_j}[z] = \Parikh_{R_i}[z] + j - i$.
Recall we have assumed that $\Parikh_{L_i} = \Parikh_{R_i}$ and 
$\Parikh_{L_j} = \Parikh_{R_j}$ for any $i < j \leq \bp(i)$.
This is possible only if $c_p = c_q = c_t$,
namely, $w[\Beg(L_j)] = w[\Beg(R_j)] = w[\End(R_j)+1]$.
\end{proof}

\begin{lemma}\label{lem:diff=2_1}
Assume $\Diff(\Parikh_{L_i}, \Parikh_{R_i}) = 2$.
Let $c_p$ be the unique character which occurs more in the left window $L_i$ than in the right window $R_i$, and
$c_q$ be the unique character which occurs more in the right window $R_i$ than in the left window $L_i$.
Let $x = \Parikh_{L_i}[p] - \Parikh_{R_i}[p] = \Parikh_{R_i}[q] - \Parikh_{L_i}[q] > 0$, and assume $x \leq \min\{D^i_1,D^i_2,D^i_3\}$.
Then, $i+x$ is the beginning position of an Abelian square of length $2d$
iff $w[\Beg(L_i)] = c_p$, $w[\Beg(R_i)] = c_q = w[\End(R_i)+1]$.
Also, this is the only Abelian square of length $2d$ beginning at positions
between $i$ and $\bp(i)$.
\end{lemma}

\begin{proof}\label{pro:diff=2_1}
  ($\Leftarrow$)
  Since $w[\Beg(L_i)] = c_p$ and $w[\Beg(R_i)] = w[\End(R_i)+1] = c_q$,
  we have that $\Parikh_{L_i}[p] - \Parikh_{R_i}[p] - z = \Parikh_{L_{i+z}}[p] - \Parikh_{R_{i+z}}[p]$ and
  $\Parikh_{R_i}[q] - \Parikh_{L_i}[q] + z = \Parikh_{R_{i+z}}[q] = \Parikh_{L_{i+z}}[q]$ for any $1 \leq z \leq \min\{D^i_1,D^i_2,D^i_3\}$.
  By the definition of $x$,
  the Parikh vectors of the sliding windows become equal at position $i+x$.
    
  ($\Rightarrow$)
  Since $x = \Parikh_{L_i}[p] - \Parikh_{R_i}[p] = \Parikh_{R_i}[q] - \Parikh_{L_i}[q] > 0$,
  $\Parikh_{L_{i+x}}[p] = \Parikh_{L_{i+x}}[p]$,
  and $\Parikh_{L_{i+x}}[q] = \Parikh_{L_{i+x}}[q]$,
  we have $w[\Beg(L_i)] = c_p$ and $w[\Beg(R_i)] = w[\End(R_i)+1] = c_q$.
  From the above arguments, it is clear that $i+x$
  is the only position between $i$ and $\bp(i)$
  where an Abelian square of length $2d$ can start.
\end{proof}

\begin{lemma}\label{lem:diff=2_2}
Assume $\Diff(\Parikh_{L_i}, \Parikh_{R_i}) = 2$.
Let $c_p$ be the unique character which occurs more in the left window $L_i$ than in the right window $R_i$, and
$c_q$ be the unique character which occurs more in the right window $R_i$ than in the left window $L_i$.
Let $x = \Parikh_{L_i}[p] - \Parikh_{R_i}[p] = \Parikh_{R_i}[q] - \Parikh_{L_i}[q] > 0$, and assume $\frac{x}{2} \leq \min\{D^i_1,D^i_2,D^i_3\}$.
Then, $i+\frac{x}{2}$ is the beginning position of an Abelian square of length $2d$
iff $w[\Beg(L_i)] = c_p = w[\End(R_i)+1]$, $w[\Beg(R_i)] = c_q$.
Also, this is the only Abelian square of length $2d$ beginning at positions
between $i$ and $\bp(i)$.
\end{lemma}

\begin{proof}\label{pro:diff=2_2}
  ($\Leftarrow$)
  Since $w[\Beg(L_i)] = c_p = w[\End(R_i)+1]$ and $w[\Beg(R_i)] = c_q$,
  we have that $\Parikh_{L_i}[p] - \Parikh_{R_i}[p] - 2z = \Parikh_{L_{i+z}}[p] - \Parikh_{R_{i+z}}[p]$ and
  $\Parikh_{R_i}[q] - \Parikh_{L_i}[q] + 2z = \Parikh_{R_{i+z}}[q] = \Parikh_{L_{i+z}}[q]$ for any $1 \leq z \leq \min\{D^i_1,D^i_2,D^i_3\}$.
  Since $\frac{x}{2} \leq \min\{D^i_1,D^i_2,D^i_3\}$,
  the Parikh vectors of the sliding windows become equal at position $i+\frac{x}{2}$.
  ($\Rightarrow$)
  Since $x = \Parikh_{L_i}[p] - \Parikh_{R_i}[p] = \Parikh_{R_i}[q] - \Parikh_{L_i}[q] > 0$,
  $\Parikh_{L_{i+\frac{x}{2}}}[p] = \Parikh_{L_{i+\frac{x}{2}}}[p]$,
  and $\Parikh_{L_{i+\frac{x}{2}}}[q] = \Parikh_{L_{i+\frac{x}{2}}}[q]$,  
  we have   $w[\Beg(L_i)] = c_p = w[\End(R_i)+1]$ and $w[\Beg(R_i)] = c_q$.
  From the above arguments, it is clear that $i+\frac{x}{2}$
  is the only position between $i$ and $\bp(i)$
  where an Abelian square of length $2d$ can start.
\end{proof}

\begin{lemma} \label{lem:diff=3}
Assume $\Diff(\Parikh_{L_i}, \Parikh_{R_i}) = 3$.
Let $c_{p} = w[\Beg(L_i)]$,
$c_{p'} = w[\End(R_i)+1]$, and
$c_q = w[\Beg(R_i)]$.
Then, $i+x$ with $i < i+x \leq \bp(i)$
is the beginning position of an Abelian square of length $2d$
iff
$0 < x = \Parikh_{L_i}[p] - \Parikh_{R_i}[p] = \Parikh_{L_i}[p'] - \Parikh_{R_i}[p'] = \frac{\Parikh_{R_i}[q] - \Parikh_{L_i}[q]}{2} \leq \min\{D_1^i, D_2^i, D_3^i\}$.
Also, this is the only Abelian square of length $2d$ beginning at positions
between $i$ and $\bp(i)$.
\end{lemma}

\begin{proof}\label{pro:diff=3}
  ($\Leftarrow$)
  Since $w[\Beg(L_i)] = c_p$,  $w[\End(R_i)+1] = c_{p'}$ and $w[\Beg(R_i)] = c_q$,
  we have that $\Parikh_{L_i}[p] - z = \Parikh_{L_{i+z}}[p]$,
  $\Parikh_{L_i}[q] + z = \Parikh_{L_{i+z}}[q]$,
  $\Parikh_{R_i}[q] - z = \Parikh_{R_{i+z}}[q]$,
  $\Parikh_{L_i}[q] + z = \Parikh_{L_{i+z}[q]}$ and
  $\Parikh_{R_i}[p'] + z = \Parikh_{R_{i+z}}[p']$
   for any $1 \leq z \leq \min\{D^i_1,D^i_2,D^i_3\}$.
  Since $x \leq \min\{D^i_1,D^i_2,D^i_3\}$,
  the Parikh vectors of the sliding windows become equal at position $i+x$
  and $i < i+x \leq \bp(i)$.

  ($\Rightarrow$)
  Since $i < i+x \leq \bp(i)$,
  we have $ < x \leq \min\{D^i_1,D^i_2,D^i_3\}$.
  Since $w[\Beg(L_i)] = c_p$, $w[\End(R_i)+1] = c_{p'}$,
  $w[\Beg(R_i)] = c_q$, and $\Parikh_{L_{i+x}} = \Parikh_{R_{i+x}}$,
  we have $x = \Parikh_{L_i}[p] - \Parikh_{R_i}[p]
  = \Parikh_{L_i}[p'] - \Parikh_{R_i}[p']
  = \frac{\Parikh_{R_i}[q] - \Parikh_{L_i}[q]}{2}$.

  From the above arguments, it is clear that $i+x$
  is the only position between $i$ and $\bp(i)$
  where an Abelian square of length $2d$ can start. 
\end{proof}

\begin{lemma}\label{lem:diff>=4}
  Assume $\Diff(\Parikh_{L_i}, \Parikh_{R_i}) \geq 4$.
  Then, there exists no Abelian square of length $2d$
  beginning at any position $j$
  with $i < j \leq \bp(i)$.
\end{lemma}

\begin{proof}\label{pro:diff>=4}
  By the definition of $\bp(i)$,
  we have that $w[\Beg(L_i)] = w[\Beg(L_{\bp(i)})-1]$,
  $w[\Beg(R_i)] = w[\Beg(R_{\bp(i)})-1]$, and
  $w[\End(R_i)] = w[\End(R_{\bp(i)})-1]$.
  Since the ending position of the left sliding window
  is adjacent to the beginning position of the right sliding window,
  we have $\Diff(\Parikh_{L_i}, \Parikh_{R_i}) - \Diff(\Parikh_{L_{j}}, \Parikh_{R_{j}}) \leq 3$ for any $i \leq j \leq \bp(i)$.
  Since we have assumed $\Diff(\Parikh_{L_i}, \Parikh_{R_i}) \geq 4$,
  we get $\Diff(\Parikh_{L_{j}}, \Parikh_{R_{j}}) \geq 1$.
  Thus there exist no Abelian squares starting at position $j$.
\end{proof}

We are ready to show the main result of this section.
\begin{theorem}
Given a string $w$ of the length $n$ over an alphabet of size $\sigma$, 
we can compute all Abelian squares in $w$ in $O(mn)$ time
and $O(n)$ working space,
where $m$ is the size of $\RLE(w)$.
\end{theorem}

\begin{proof}
Consider an arbitrarily fixed window length $d$.
As was explained, it takes $O(m)$ time to compute
$\Parikh_{L_1}$, $\Parikh_{R_1}$, and $\Diff(\Parikh_{L_1}, \Parikh_{R_1})$
for the initial position $1$.
Suppose that the two windows are aligned at some position $i \geq 1$.
Then, our algorithm computes Abelian squares starting
at positions between $i$ and $\bp(i)$
using one of Lemma~\ref{lem:diff=0}, 
Lemma~\ref{lem:diff=2_1}, Lemma~\ref{lem:diff=2_2},
Lemma~\ref{lem:diff=3}, and Lemma~\ref{lem:diff>=4},
depending on the value of $\Diff(\Parikh_{L_1}, \Parikh_{R_i})$.
In each case, all Abelian squares of length $2d$
starting at positions between $i$ and $\bp(i)$ can be computed 
in $O(1)$ time by simple arithmetics.
Then, the left and right windows $L_i$ and $R_i$
are shifted to $L_{\bp(i)}$ and $R_{\bp(i)}$,
respectively.
Using the array $S$ as in Theorem~\ref{theo:Abelian_periods},
we can compute $\bp(i)$ in $O(1)$ time
for a given position $i$ in $w$.

Let us analyze the number of times the windows are shifted for each $d$.
Since $\bp(i) = i + \min\{D^i_1,D^i_2,D^i_3\}$,
for each position $p$ there can be at most three distinct positions 
$i, j, k$ such that $p = \bp(i) = \bp(j) = \bp(k)$.
Thus, for each $d$ we shift the two adjacent windows at most $3m$ times.

Overall, our algorithm runs in $O(mn)$ time for all window lengths
$d = 1, \ldots, \lfloor n/2 \rfloor$.
The space requirement is $O(n)$ since 
we need to maintain the Parikh vectors of the two sliding windows
and the array $S$.
\end{proof}


\subsection{Example for Computing Abelian squares using RLEs}
\label{square_examples}

Here we show some examples
on how our algorithm 
computes all Abelian squares of a given string based on its RLE.

Consider string $w=a^{12}b^4a^3c^2d^2c^2a^2$
over alphabet $\Sigma = \{a,b,c,d\}$ of size $4$.
Let $d=4$.

\begin{figure}[!h]
\centering{
\includegraphics[width=9cm]{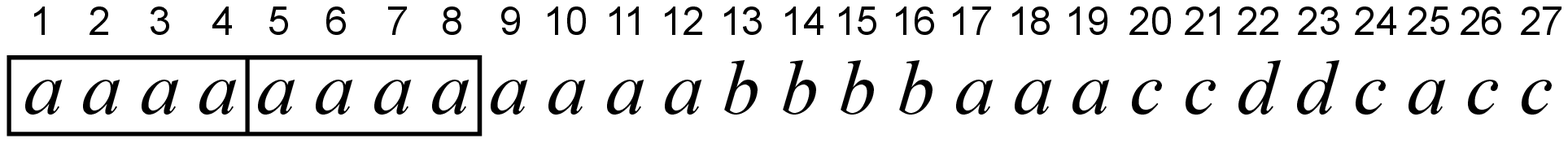}
}
\caption{$\Beg(L_1) = 1, \Beg(R_1) = 5, \End(R_1)+1 = 9, w[\Beg(L_1)] = w[\Beg(R_1)] = w[\End(R_1)+1] = a$.}
\label{fig:step1}
\end{figure}
See Figure~\ref{fig:step1} for the initial step of our algorithm,
where $i = 1$.
As $\Diff(\Parikh_{L_1}, \Parikh_{R_1}) = 0$,
$w[1..8] = aaaaaaaa$ is an Abelian square.
Since $\min \{D_1^1,D_2^1,D_3^1\} = \min \{12,8,4\} = 4$,
the next break point is $\bp(1) = 1 + 4 = 5$.
Since $w[\Beg(L_1)] = w[\Beg(R_1)] = w[\End(R_1)+1] = a$ and
it follows from Lemma~\ref{lem:diff=0} that
the substrings of length $2d = 8$
between $1$ and the break point are all equal,
i.e., $w[1..8] = w[2..9] = w[3..10] = w[4..11] = w[5..12]$,
and all of them are Abelian squares.
Hence we output a triple $\langle 1, 5, 4 \rangle$
representing all these Abelian squares.
We update $i \leftarrow \bp(1) = 5$,
and proceed to the next step.

\begin{figure}[!h]
\centering{
\includegraphics[width=9cm]{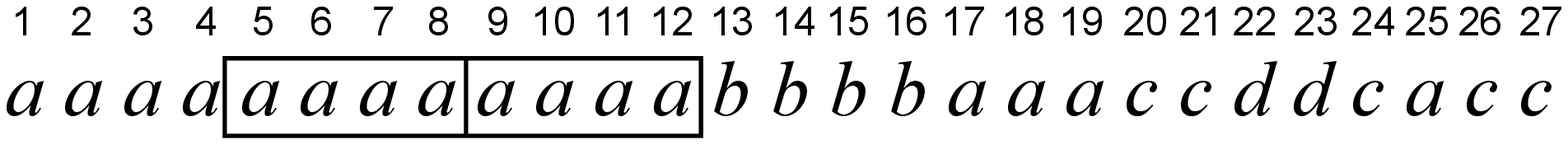}
}
\caption{$\Beg(L_5) = 5, \Beg(R_5) = 9, \End(R_5)+1 = 13, w[\Beg(L_5)] = w[\Beg(R_5)] = a, w[\End(R_5)+1] = b$.}
\label{fig:step2}
\end{figure}
Next, see Figure~\ref{fig:step2} 
where the left window has been shifted
to $L_5 = w[5..6] = aaaa$ and the right window has been shifted to
$R_5 = w[8..12] = aaaa$.
Since $\min \{D_1^5,D_2^5,D_3^5\} = \min \{8,4,4\} = 4$,
the next break point is $\bp(5) = 5+4 = 9$.
Since $\Parikh_{L_5} = \Parikh_{R_5}$ and
$w[\Beg(L_5)] = w[\Beg(R_5)] = a \neq w[\End(R_5)+1] = b$,
it follows from Lemma~\ref{lem:diff=0} that 
there are no Abelian squares between $5$
and the break point $9$.
We update $i \leftarrow \bp(5) = 9$,
and proceed to the next step.

\begin{figure}[!h]
\centering{
\includegraphics[width=9cm]{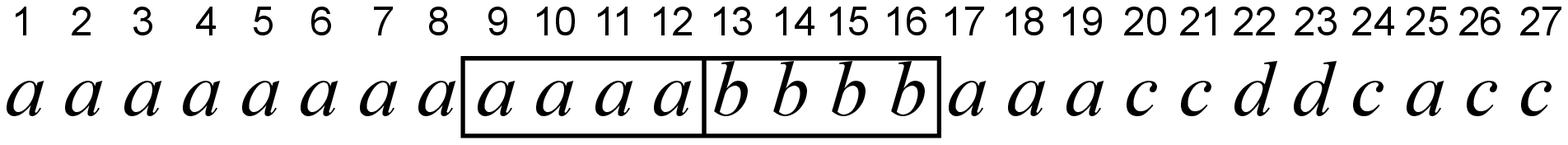}
}
\caption{$\Beg(L_9) = 9, \Beg(R_9) = 13, \End(R_9)+1 = 17, w[\Beg(L_9)] = a, w[\Beg(R_9)] = b, w[\End(R_9)+1] = a$.}
\label{fig:step3}
\end{figure}
Next, see Figure~\ref{fig:step3}
where the left window has been shifted
to $L_9 = w[9..12] = aaaa$ and the right window has been shifted to
$R_9 = w[13..16] = bbbb$.
Since $\min \{D_1^9,D_2^9,D_3^9\} = \min \{4,4,3\} = 3$,
the next break point is $\bp(9) = 9+3 = 12$.
Since $\Diff(\Parikh_{L_9}, \Parikh_{R_9}) = 2$,
$w[\Beg(L_9)] = w[\End(R_9)+1] = a \neq w[\Beg(R_9)] = b$,
and $\frac{\Parikh_{L_9}[a] - \Parikh_{R_9}[a] = \Parikh_{R_9}[b] - \Parikh_{L_9}[b]}{2} = 2 \leq \min \{D_1^9,D_2^9,D_3^9\} = 3$,
it follows from Lemma~\ref{lem:diff=2_2} that
$w[11..18]$ is the only Abelian square of length $2d = 8$
starting at positions between $9$ and $12$.
We hence output $\langle 11, 11, 4 \rangle$.
We update $i \leftarrow \bp(9) = 12$,
and proceed to the next step.

\begin{figure}[!h]
\centering{
\includegraphics[width=9cm]{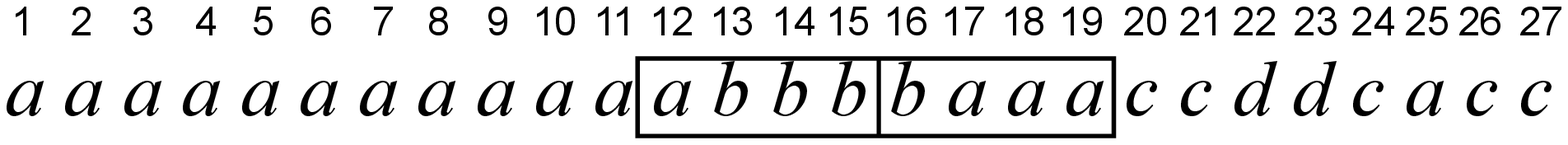}
}
\caption{$\Beg(L_{12}) = 12, \Beg(R_{12}) = 16, \End(R_{12})+1 = 20, w[\Beg(L_{12})] = a, w[\Beg(R_{12})] = b, w[\End(R_{12})+1] = c$.}
\label{fig:step4}
\end{figure}
Next, see Figure~\ref{fig:step4}
where the left window has been shifted
to $L_{12} = w[12..15] = abbb$ and the right window has been shifted to
$R_{12} = w[16..19] = baaa$.
Since $\min \{D_1^{12},D_2^{12},D_3^{12}\} = \min\{1, 1, 1\} = 1$,
the next break point is $\bp(12) = 12+1 = 13$.
Since $\Diff(\Parikh_{L_{12}}, \Parikh_{R_{12}}) = 3$ and
$w[\Beg(L_{12})] = a \neq w[\Beg(R_{12})] = b \neq w[\End(R_{12})+1] = c$,
it follows from Lemma~\ref{lem:diff=2_1} and Lemma~\ref{lem:diff=2_2} that 
there are no Abelian squares starting at positions between $12$
and $13$.
We update $i \leftarrow \bp(12) = 13$,
and proceed to the next step.

\begin{figure}[!h]
\centering{
\includegraphics[width=9cm]{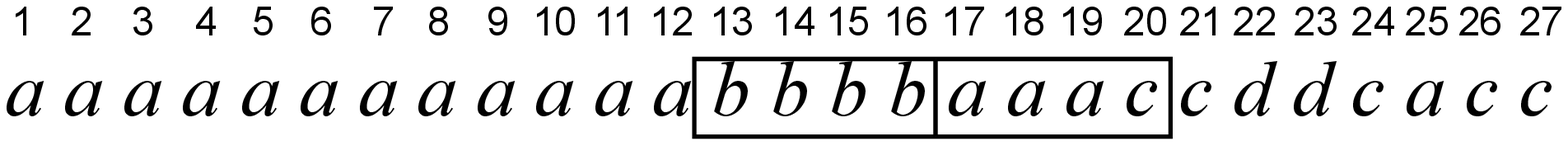}
}
\caption{$\Beg(L_{13}) = 13, \Beg(R_{13}) = 17, \End(R_{13})+1 = 21, w[\Beg(L_{13})] = b, w[\Beg(R_{13})] = a, w[\End(R_{13})+1] = c$.}
\label{fig:step5}
\end{figure}
Next, see Figure~\ref{fig:step5}
where the left window has been shifted
to $L_{13} = w[13..16] = bbbb$ and the right window has been shifted to
$R_{13} = w[17..20] = aaac$.
Since $\min\{D_1^{13}, D_2^{13}, D_3^{13}\} = \min \{4, 3, 1\} = 1$,
the next break point is $\bp(13) = 13+1 = 14$.
Since $\Diff(\Parikh_{L_{13}}, \Parikh_{R_{13}}) = 3$
and $\Parikh_{L_{13}}[b] - \Parikh_{R_{13}}[b] = 4 \neq -1 = \Parikh_{L_{13}}[c] - \Parikh_{R_{13}}[c]$,
it follows from Lemma~\ref{lem:diff=3} that 
$14$ is not the beginning position of an Abelian square of length $2d = 8$.
We update $i \leftarrow \bp(13) = 14$,
and proceed to the next step.

\begin{figure}[!h]
\centering{
\includegraphics[width=9cm]{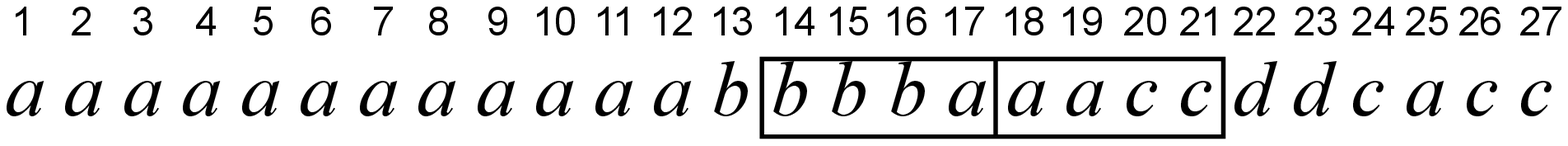}
}
\caption{$\Beg(L_{14}) = 14, \Beg(R_{14}) = 18, \End(R_{14})+1 = 22, w[\Beg(L_{14})] = b, w[\Beg(R_{14})] = a, w[\End(R_{14})+1] = d$.}
\label{fig:step6}
\end{figure}
Next, see Figure~\ref{fig:step6}
where the left window has been shifted
to $L_{14} = w[14..17] = bbba$ and the right window has been shifted to
$R_{14} = w[18..21] = aacc$.
Since $\min\{D_1^{14}, D_2^{14}, D_3^{14}\} = \min \{3, 2, 2\} = 2$,
the next break point is $\bp(14) = 14+2 = 16$.
Since $\Diff(\Parikh_{L_{14}}, \Parikh_{R_{14}}) = 3$
and $\Parikh_{L_{14}}[b] - \Parikh_{R_{14}}[b] = 3 \neq -1 = \Parikh_{L_{14}}[c] - \Parikh_{R_{14}}[c]$,
it follows from Lemma~\ref{lem:diff=3} that 
there are no Abelian squares starting at positions between $14$ and $16$.
We update $i \leftarrow \bp(14) = 16$,
and proceed to the next step.

\begin{figure}[!h]
\centering{
\includegraphics[width=9cm]{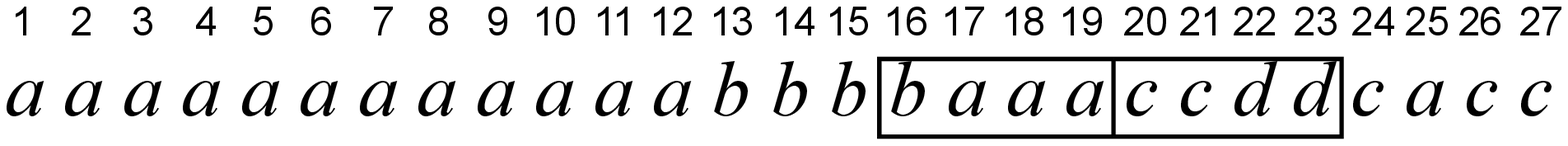}
}
\caption{$\Beg(L_{16}) = 16, \Beg(R_{16}) = 20, \End(R_{16})+1 = 24, w[\Beg(L_{16})] = b, w[\Beg(R_{16})] = w[\End(R_{16})+1] = c$}
\label{fig:step7}
\end{figure}
Next, see Figure~\ref{fig:step7}
where the left window has been shifted
to $L_{16} = w[16..19] = baaa$ and the right window has been shifted to
$R_{16} = w[20..23] = ccdd$.
Since $\min\{D_1^{16}, D_2^{16}, D_2^{16}\} = \min \{1, 2, 2\} = 1$,
the next break point is $\bp(16) = 16+1 =17$.
Since $\Diff(\Parikh_{L_{16}}, \Parikh_{R_{16}}) = 3$
and $\Parikh_{L_{16}}[b] - \Parikh_{R_{16}}[b] = 1 \neq -2 = \Parikh_{L_{16}}[c] - \Parikh_{R_{16}}[c]$,
it follows from Lemma~\ref{lem:diff=3} that 
$16$ is not the beginning position of an Abelian square of length $2d = 8$.
We update $i \leftarrow \bp(16) = 17$,
and proceed to the next step.

\begin{figure}[!h]
\centering{
\includegraphics[width=9cm]{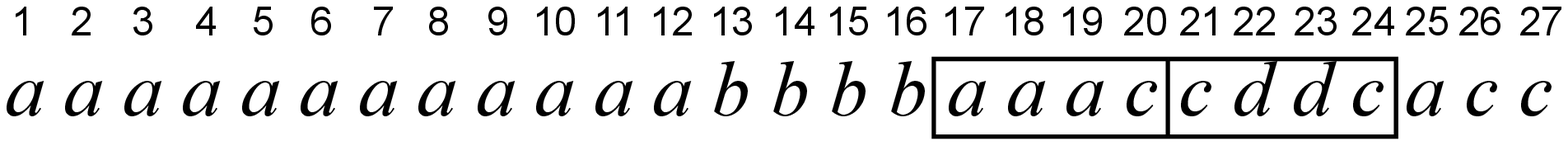}
}
\caption{$\Beg(L_{17}) = 17, \Beg(R_{17}) = 21, \End(R_{17})+1 = 25, w[\Beg(L_{17})] = a, w[\Beg(R_{17})] = c, w[\End(R_{17})+1] = a$}
\label{fig:step8}
\end{figure}
Next, see Figure~\ref{fig:step8}
where the left window has been shifted
to $L_{17} = w[17..20] = aaac$ and the right window has been shifted to
$R_{17} = w[21..24] = cddc$.
Since $\min\{D_1^{17}, D_2^{17}, D_2^{17}\} = \min \{3, 1, 1\} = 1$,
the next break point is $\bp(17) = 17+1 =18$.
Since $\Diff(\Parikh_{L_{17}}, \Parikh_{R_{17}}) = 3$
and $\Parikh_{L_{17}}[a] - \Parikh_{R_{17}}[a] = 3 \neq -2 = \Parikh_{L_{17}}[c] - \Parikh_{R_{17}}[c]$,
it follows from Lemma~\ref{lem:diff=3} that 
$17$ is not the beginning position of an Abelian square of length $2d = 8$.
We update $i \leftarrow \bp(17) = 18$,
and proceed to the next step.

\begin{figure}[!h]
\centering{
\includegraphics[width=9cm]{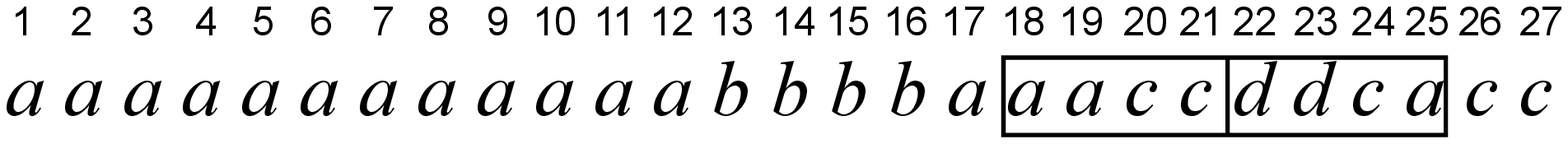}
}
\caption{$\Beg(L_{18}) = 18, \Beg(R_{18}) = 22, \End(R_{18})+1 = 26, w[\Beg(L_{18})] = a, w[\Beg(R_{18})] = d, w[\End(R_{18})+1] = c$}
\label{fig:step9}
\end{figure}
Next, see Figure~\ref{fig:step9}
where the left window has been shifted
to $L_{18} = w[18..21] = aacc$ and the right window has been shifted to
$R_{18} = w[20..25] = ddcc$.
Since $\min\{D_1^{18}, D_2^{18}, D_2^{18}\} = \min \{2, 2, 2\} = 2$,
the next break point is $\bp(18) = 18+2 =20$.
Since $\Diff(\Parikh_{L_{18}}, \Parikh_{R_{18}}) = 3$,
we use Lemma~\ref{lem:diff=3}.
Since $\Parikh_{L_{18}}[a] - \Parikh_{R_{18}}[a] = \Parikh_{L_{18}}[c] - \Parikh_{R_{18}}[c] = \frac{\Parikh_{R_{18}}[d] - \Parikh_{L_{18}}[d]}{2} = 1 \leq \min \{D_1^{18},D_2^{18},D_3^{18}\} = 2$,
it follows from Lemma~\ref{lem:diff=3} that
$w[19..26]$ is an Abelian square of length $2d = 8$.
We hence output $\langle 19, 19, 4 \rangle$.
We update $i \leftarrow \bp(19) = 20$,
and proceed to the next step.

\begin{figure}[!h]
\centering{
\includegraphics[width=9cm]{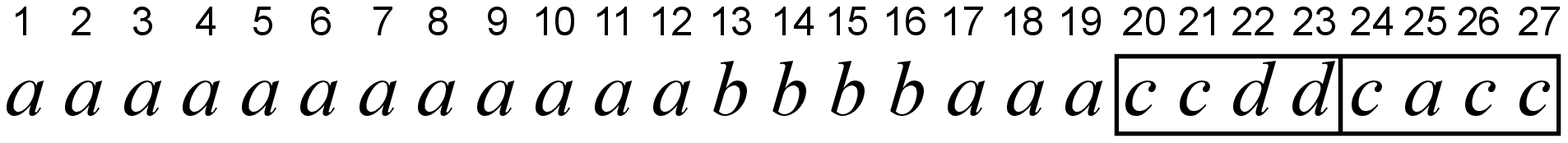}
}
\caption{$\Beg(L_{20}) = 20, \Beg(R_{20}) = 24, w[\Beg(L_{20})] = c, w[\Beg(R_{20})] = c$}
\label{fig:step10}
\end{figure}
Next, see Figure~\ref{fig:step10}
where the left window has been shifted
to $L_{20} = w[20..23] = ccdd$ and the right window has been shifted to
$R_{20} = w[24..27] = cacc$.
Since $\Diff(\Parikh_{L_{20}}, \Parikh_{R_{20}}) = 3$
the right end of the right window has reached
the last positions of the input string,
the algorithm terminates here.
Recall that this algorithm computed all the Abelian squares of length $2d = 8$
in this string.

\section{Computing longest common Abelian factors using RLEs}
\label{sec:lcaf}

In this section, we introduce our RLE-based
algorithm which computes longest common Abelian factors of two given strings
$w_1$ and $w_2$.

Formally, we solve the following problem.
Let $n = \min\{|w_1|, |w_2|\}$.
Given two strings $w_1$ and $w_2$, compute the length
$l = \max\{d \in [1, n] \mid 1 \leq \exists i \leq |w_1|, 1 \leq \exists k \leq |w_2|
\mbox{ s.t. } \Parikh_{w_1[i..i+d-1]} = \Parikh_{w_2[k..k+d-1]}\}$
of the longest common Abelian factor(s) of $w_1$ and $w_2$,
together with a pair $(i, k)$ of positions on $w_1$ and $w_2$
such that $\Parikh_{w_1[i..i+l-1]} = \Parikh_{w_2[k..k+l-1]}$.

\subsection{Alatabbi et al.'s $O(\sigma n^2)$-time algorithm}
Our algorithm uses an idea from
Alattabi et al.'s algorithm~\cite{alatabbi:_algor_longes_common_abelian_factor}.

For each window size $d$,
their algorithm computes the Parikh vectors of
all substrings of $w_1$ and $w_2$ of length $d$ in $O(\sigma n)$ time,
using two windows of length $d$ each.
Then they sort the Parikh vectors in $O(\sigma n)$ time,
and output the largest $d$ for which common Parikh vectors exist
for $w_1$ and $w_2$,
together with the lists of respective occurrences of longest common Abelian factors.

The total time requirement is clearly $O(\sigma n^2)$.

\subsection{Our $O(m^2n)$-time algorithm}
Our algorithm is different from Alattabi et al.'s algorithm in that
(1) we use RLEs of strings $w_1$ and $w_2$ and
(2) we avoid to sort the Parikh vectors.

As in the previous sections,
for a given window length $d~(1 \leq n)$,
we shift two windows of length $d$ over both $\RLE(w_1)$ and $\RLE(w_2)$,
and stops when we reach a break point of $\RLE(w_1)$ or $\RLE(w_2)$.
We then check if there is a common Abelian factor in the ranges of $w_1$ and $w_2$ we are looking at.

Since we use a single window for each of the input strings $w_1$ and $w_2$,
we need to modify the definition of the break points.
Let $U_i$ and $V_k$ be the sliding windows for $w_1$ and $w_2$
that are aligned at position $i$ of $w_1$ and at position $k$ of $w_2$,
respectively.
For each position $i \geq 1$ in $w_1$,
let $\bp_1(i) = i+ \min\{D_1^i, D_2^i\}$,
where $D_1^i = \Succ(\Beg(U_i)) - i$ and $D_2^i = \Succ(\End(U_i)) - i$.
For each position $k \geq 1$ in $w_2$,
$\bp_2(k)$ is defined analogously.
Let $p_l = \Beg(U_i)$, $p_r = \End(U_i)+1$, $q_l = \Beg(V_k)$ and $q_r = \End(V_k)+1$.

Consider an arbitrarily fixed window length $d$.
Assume that we have just shifted the window on $w_1$ from
position $i$ (i.e., $U_i = w_1[i..i+d-1]$) to 
the break point $\bp_1(i)$ (i.e., $U_{\bp_1(i)} = w_1[\bp_1(i)..\bp_1(i)+d-1]$).
Let $c_{p_l} = w_1[i]$ and $c_{p_r} = w_1[i+d]$ (see also Figure~\ref{fig:lcaf}).

\begin{figure}[t]
\centerline{
\includegraphics[width=7cm]{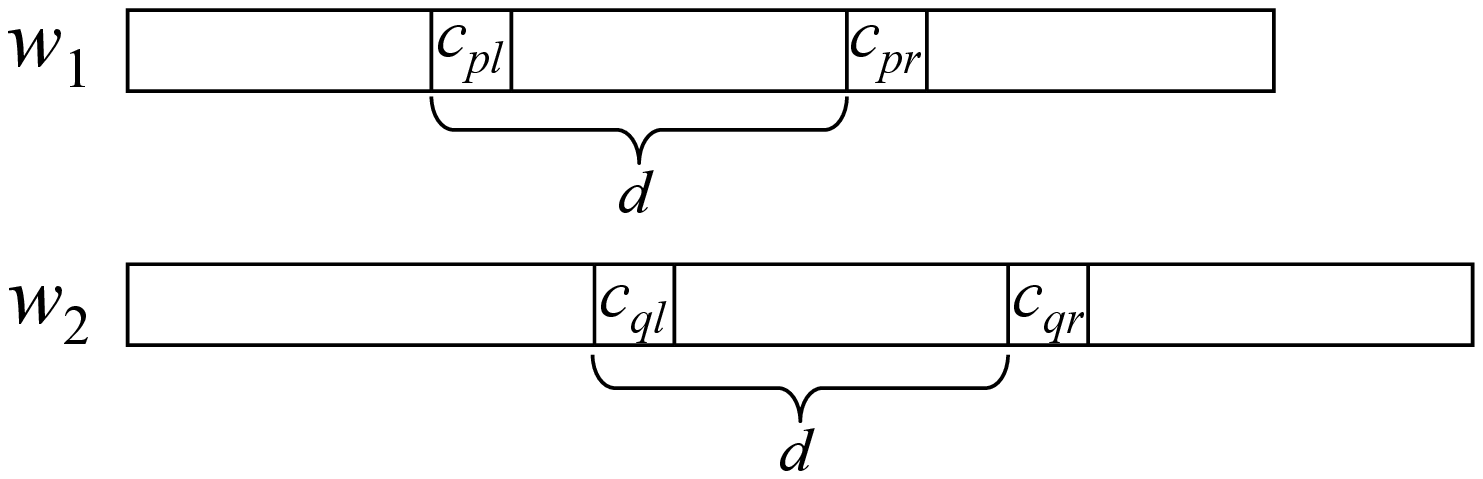}
}
\caption{Conceptual drawing of $c_{p_l}$, $c_{p_r}$, $c_{q_r}$, and $c_{q_l}$.}
\label{fig:lcaf}
\end{figure}

For characters $c_{p_l}$ and $c_{p_r}$,
we consider the minimum and maximum numbers of their occurrences 
during the slide from position $i$ to $\bp_1(i)$.
Let $min(p_l) = \Parikh_{w_1[\bp_1(i)..\bp_1(i)+d-1]}[p_l]$, $max(p_l) = \Parikh_{w_1[i..i+d-1]}[p_l]$, $min(p_r) = \Parikh_{w_1[i..i+d-1]}[p_r]$ and $max(p_r) = \Parikh_{w_1[\bp_1(i)..\bp_1(i)+d-1]}[p_r]$.
We will use these values to determine if there is a common Abelian factor
of length $d$ for $w_1$ and $w_2$.

Also, assume that we have just shifted the window on $w_2$ from
position $k$ (i.e., $V_k = w_2[k..k+d-1]$) to 
the break point $\bp_2(k)$ (i.e., $V_{\bp_2(k)} = w_2[\bp_2(k)..\bp_2(k)+d-1]$).

Let $c_{q_l} = w_2[k]$ and $c_{q_r} = w_2[k+d]$ (see also Figure~\ref{fig:lcaf}).
For characters $c_{q_l}$ and $c_{q_r}$,
we also consider the minimum and maximum numbers of occurrences of
of these characters during the slide from position $k$ to $\bp_2(k)$.
Let $min(q_l) = \Parikh_{w_2[\bp_2(k)..\bp_2(k)+d-1]}[q_l]$, $max(q_l) = \Parikh_{w_2[k..k+d-1]}[q_l]$, $min(q_r) = \Parikh_{w_2[k..k+d-1]}[q_r]$ and $max(q_r) = \Parikh_{w_2[\bp_2(k)..\bp_2(k)+d-1]}[q_r]$.

Let $m$ be the total size of $\RLE(w_1)$ and $\RLE(w_2)$,
and $l$ be the length of longest common Abelian factors of $w_1$ and $w_2$.
Our algorithm computes an $O(m^2)$-size representation
of every pair $(i, k)$ of positions for which 
$(w_1[i..i+l-1], w_2[k..k+l-1])$ is a longest common Abelian factor
of $w_1$ and $w_2$.

In the lemmas which follow, 
we assume that $\Parikh_{w_1[i..i+d-1]}[v] = \Parikh_{w_2[k..k+d-1]}[v]$ for any $v \in \{1,..,\sigma\} \setminus \{p_l,p_r,q_l,q_r\}$.
This is because, if this condition is not satisfied, then there cannot be 
an Abelian common factor of length $d$
for positions between $i$ to $\bp_1(i)$ in $w_1$
and position between $k$ to $\bp_2(k)$ in $w_2$.

\begin{lemma}\label{lem:all}
Assume $c_{p_l} = c_{p_r}$ and $c_{q_l} = c_{q_r}$.
Then, for any pair of positions $i \leq i' \leq \bp_1(i)$
and $k \leq k' \leq \bp_2(k)$,
$(w_1[i'..i'+d-1], w_2[k'..k'+d-1])$ is an Abelian common factor of length $d$
iff $\Parikh_{w_1[i..i+d-1]} = \Parikh_{w_2[k..k+d-1]}$.
\end{lemma}

\begin{proof}
Since $c_{p_l} = c_{p_r}$ and $c_{q_l} = c_{q_r}$,
the Parikh vectors of the sliding windows do not change during the slides
from $i$ to $\bp_1(i)$ and from $k$ to $\bp_2(k)$.
Thus the lemma holds.
\end{proof}

\begin{lemma}\label{lem:II}
Assume $c_{p_l} = c_{q_l} \neq c_{p_r} = c_{q_r}$.
There is a common Abelian common factor $(w_1[i+x..i+x+d-1], w_2[k+y..k+y+d-1])$ of length $d$ iff 
$0 \leq x \leq \bp_1(i)-i$, $0 \leq y \leq \bp_2(k)-k$ and
$x-y = max(p_l)-max(q_l) = min(q_r)-min(p_r)$.
\end{lemma}

\begin{proof}\label{pro:II}
During the slide of the window on $w_1$,
the number of occurrences of $c_{p_l}$ decreases and
that of $c_{p_r}$ increases.
That is, $\Parikh_{w_1[i+x..i+x+d-1]}[p_l] = \Parikh_{w_1[i..i+d-1]}[p_l] - x = max(p_l) - x$ and $\Parikh_{w_1[i+x..i+x+d-1]}[p_r] = \Parikh_{w_1[i..i+d-1]}[p_r] + x = min(p_r) + x$.
On the other hand, during the slide of the window on $w_2$,
the number of occurrence of $c_{q_l}$ decreases and that of $c_{q_r}$ increases.
That is, $\Parikh_{w_2[k+y..k+y+d-1]}[q_l] = \Parikh_{w_2[k..k+d-1]}[q_l] - y = max(q_l) - y$ and $\Parikh_{w_2[k+y..k+y+d-1]}[p_r] = \Parikh_{w_2[k..k+d-1]}[q_r] + y = min(q_r) + y$.

Assume a pair $(w_1[i+x..i+x+d-1] , w_2[k+y..k+y+d-1])$ is a common Abelian factor of length $d$.
Then, $\Parikh_{w_1[i+x..i+x+d-1]}[p_l] = \Parikh_{w_2[k+y..k+y+d-1]}[q_l]$ and $\Parikh_{w_1[i+x..i+x+d-1]}[p_r] = \Parikh_{w_2[k+y..k+y+d-1]}[q_r]$, that is, $max(p_l) - x = max(q_l) \\ - y$ and $min(p_r) + x = min(q_r) + y$. Therefore $x-y = max(p_l)-max(q_l) = min(q_r)-min(p_r)$.

Assume that $x-y = max(p_l)-max(q_l) = min(q_r)-min(p_r)$.
Then, we have that $max(p_l) - max(q_l) = \Parikh_{w_1[i..i+d-1]}[p_l] - \Parikh_{w_2[k..k+d-1]}[q_l] = \Parikh_{w_1[i+x..i+x+d-1]}[p_l]+x - \Parikh_{w_2[k+y..k+y+d-1]}[q_l]-y = x-y$,
that is, $\Parikh_{w_1[i+x..i+x+d-1]}[p_l] = \Parikh_{w_2[k+y..k+y+d-1]}[q_l]$.
Also, we have that $min(q_r) - min(p_r) = \Parikh_{w_2[k..k+d-1]}[q_r] - \Parikh_{w_1[i..i+d-1]}[p_r] = \Parikh_{w_2[k+y..k+y+d-1]}[q_r]-y - \Parikh_{w_1[i+x..i+x+d-1]}[p_r]+x = x-y$, that is, $\Parikh_{w_2[k+y..k+y+d-1]}[q_r] = \Parikh_{w_1[i+x..i+x+d-1]}[p_r]$.
Therefore, a pair $(w_1[i+x..i+x+d-1] , w_2[k+y..k+y+d-1])$ is a common Abelian factor of $w_1$ and $w_2$.
\end{proof}

\begin{lemma}\label{lem:left}
Assume $c_{p_r} \neq c_{p_l} = c_{q_l} \neq c_{q_r}$ and $c_{p_r} \neq c_{q_r}$.
There is a common Abelian factor $(w_1[i+x..i+x+d-1] , w_2[k+y..k+y+d-1])$ of length $d$iff $\bp_1(i)-i \geq x = \Parikh_{w_2[k..k+d-1]}[p_r] - min(p_r) \geq 0$, $\bp_2(k)-k \geq y = \Parikh_{w_1[i..i+d-1]}[q_r] - min(q_r) \geq 0$ and $\Parikh_{w_1[i..i+d-1]}[p_l] - x = \Parikh_{w_2[k..k+d-1]}[q_l] - y$.
\end{lemma}

\begin{proof}\label{pro:left}
During the slides of the windows on $w_1$ and $w_2$,
the numbers of occurrences of $c_{q_r}$ in $w_1$ and $c_{p_r}$ in $w_2$
do not change.

Assume there is a common Abelian factor $(w_1[i+x..i+x+d-1] , w_2[k+y..k+y+d-1])$ of length $d$.
Clearly $0 \leq x \leq \bp_1(i)-i$ and $0 \leq y \leq \bp_2(k)-k$.
Then, we have
$\Parikh_{w_1[i+x..i+x+d-1]}[p_r] = \Parikh_{w_2[k+y..k+y+d-1]}[p_r]$, $ \Parikh_{w_2[k+y..k+y+d-1]}[q_r] = \Parikh_{w_1[i+x..i+x+d-1]}[q_r]$ and $\Parikh_{w_1[i+x..i+x+d-1]}[p_l]= \Parikh_{w_2[k+y..k+y+d-1]}[q_l]$, that is, $min(p_r) + x = \Parikh_{w_2[k..k+d-1]}[p_r]$, $min(q_r) + y = \Parikh_{w_1[i..i+d-1]}[q_r]$ and $\Parikh_{w_1[i..i+d-1]}$ $[p_l] - x = \Parikh_{w_2[k..k+d-1]}[q_l] - y$.
Consequently, we obtain $x = \Parikh_{w_2[k..k+d-1]}[p_r] - min(p_r)$ and $y = \Parikh_{w_1[i..i+d-1]}[q_r] - min(q_r)$.

Assume that $\bp_1(i)-i \geq x = \Parikh_{w_2[k..k+d-1]}[p_r] - min(p_r) \geq 0$, 
$\bp_2(k)-k \geq y = \Parikh_{w_1[i..i+d-1]}[q_r] - min(q_r) \geq 0$ 
and $\Parikh_{w_1[i..i+d-1]}[p_l] - x = \Parikh_{w_2[k..k+d-1]}[q_l] - y$.
Then, we have that 
$min (p_r) + x = \Parikh_{w_1[i..i+d-1]}[p_r] + x = \Parikh_{w_1[i+x..i+x+d-1]}[p_r] = \Parikh_{w_2[k..k+d-1]}[p_r] = \Parikh_{w_2[k+y..k+y+d-1]}[p_r]$,
$min(q_r) + y = \Parikh_{w_2[k..k+d-1]}[q_r] + y = \Parikh_{w_2[k+y..k+y+d-1]}[q_r] =\Parikh_{w_1[i..i+d-1]}[q_r] = \Parikh_{w_1[i+x..i+x+d-1]}[q_r]$ 
and $\Parikh_{w_1[i+x..i+x+d-1]}[p_l] = \Parikh_{w_2[k+y..k+y+d-1]}[q_l]$.
Therefore, a pair $(w_1[i+x..i+x+d-1] , w_2[k+y..k+y+d-1])$ is a common Abelian factor of $w_1$ and $w_2$.
\end{proof}

\begin{lemma}\label{lem:right}
Assume $c_{p_l} \neq c_{p_r} = c_{q_r} \neq c_{q_l}$ and $c_{p_l} \neq c_{q_l}$.
There is a common Abelian factor $(w_1[i+x..i+x+d-1] , w_2[k+y..k+y+d-1])$ of length $d$ iff $x = max(p_l) - \Parikh_{w_2[k..k+d-1]}[p_l] \geq 0$, $y = max(q_l) - \Parikh_{w_1[i..i+d-1]}[q_l] \geq 0$ and $\Parikh_{w_1[i..i+d-1]}[p_r] + x = \Parikh_{w_2[k..k+d-1]}[q_r] + y$.
\end{lemma}
Lemma~\ref{lem:right} can be proved by a similar argument to the proof of Lemma~\ref{lem:left}.

\begin{lemma}\label{lem:x}
Assume $c_{p_l} = c_{q_r} \neq c_{p_r} = c_{q_l}$.
There is a common Abelian factor $(w_1[i+x..i+x+d-1] , w_2[k+y..k+y+d-1])$
of length $d$ iff $x + y = min(p_r) - max(q_l) = max(q_l) - min(p_r)$, $0 \leq x \leq \bp_1(i)-i$ and $0 \leq y \leq \bp_2(k)-k$.
\end{lemma}

\begin{proof}\label{pro:x}
When the window on $w_1$ slides by $x$ positions,
the occurrence of $c_{p_l}$ in the window decreases by $x$ and
the occurrence of $c_{p_r}$ in the window increases by $x$.
When the window on $w_2$ slides by $y$ positions,
the occurrence of $c_{q_l}$ in the window decreases by $y$
and the occurrence of $c_{q_r}$ in the window increases by $y$.

Assume there is a common Abelian factor $(w_1[i+x..i+x+d-1] , w_2[k+y..k+y+d-1])$.
Then $\Parikh_{w_1[i+x..i+x+d-1]}[p_r] = \Parikh_{w_1[i..i+d-1]}[p_r] + x = min(p_r) + x$,
$\Parikh_{w_2[k+y..k+y+d-1]}[q_l] = \Parikh_{w_2[k..k+d-1]}[q_l] - y = max(q_l) - y$,
$\Parikh_{w_1[i+x..i+x+d-1]}[p_l] = \Parikh_{w_1[i..i+d-1]}[p_l] - x = max(p_l) - x$ and 
$\Parikh_{w_2[k+y..k+y+d-1]}[q_r] = \Parikh_{w_2[k..k+d-1]}[q_r] + y = min(q_r) + y$.
Therefore $\Parikh_{w_1[i+x..i+x+d-1]}[p_r] = \Parikh_{w_2[k+y..k+y+d-1]}[q_l] \Leftrightarrow x + y = max(q_l) - min(p_r)$
and 
$\Parikh_{w_1[i+x..i+x+d-1]}[p_l]$ $=\Parikh_{w_2[k+y..k+y+d-1]}[q_r] \Leftrightarrow x + y = max(p_l) - min(q_r)$.

Assume $x + y = max(q_l) - min(p_r) = max(p_l) - min(q_r)$.
Clearly $0 \leq x \leq \bp_1(i)-i$ and $0 \leq y \leq \bp_2(k)-k$.
Then $max(q_l) - y = min(p_r) + x$ and $max(p_l)-x = min(q_r)+y$, 
that is, $\Parikh_{w_2[k+y..k+y+d-1]}[q_l] = \Parikh_{w_1[i+x..i+x+d-1]}[p_r]$ and $\Parikh_{w_1[i+x..i+x+d-1]}[p_l] = \Parikh_{w_2[k+y..k+y+d-1]}[q_r]$.
Therefore a pair $(w_1[i+x..i+x+d-1] , w_2[k+y..k+y+d-1])$ is a common Abelian factor of $w_1$ and $w_2$.
\end{proof}

\begin{lemma}\label{lem:barabara}
Assume $c_{p_l}$, $c_{p_r}$, $c_{q_l}$ and $c_{q_r}$ are mutually distinct.
There is a common Abelian factor $(w_1[i+x..i+x+d-1] , w_2[k+y..k+y+d-1])$ of length $d$ iff $0 \leq x = max(p_l) - \Parikh_{w_2[k..k+d-1]}[p_l] = \Parikh_{w_2[k..k+d-1]}[p_r] - min(p_r) \leq \bp_1(i)-i$ and $0 \leq y = max(q_l) - \Parikh_{w_1[i..i+d-1]}[q_l] = \Parikh_{w_1[i..i+d-1]}[q_r] - min(q_r) \leq \bp_2(k)-k$.
\end{lemma}

\begin{proof}\label{pro:barabara}
  During the slides,
  the numbers of occurrences of $c_{q_l}$ and $c_{q_r}$ in the window on $w_1$
  do not change,
  and those of $c_{p_l}$ and $c_{p_r}$ in the window on $w_2$ do not change.

  Assume there is a common Abelian factor $(w_1[i+x..i+x+d-1], w_2[k+y..k+y+d-1])$.
Then,
$\Parikh_{w_1[i+x..i+x+d-1]}[p_r] = min(p_r) + x = \Parikh_{w_2}[p_r]$,
$\Parikh_{w_1[i+x..i+x+d-1]}$ $[p_l] = max(p_l) - x = \Parikh_{w_2}[p_l]$,
$\Parikh_{w_2[k+y..k+y+d-1]}[q_r] = min(q_r) + y = \Parikh_{w_1}[q_r]$ and
$\Parikh_{w_2[k+y..k+y+d-1]}[q_l] = max(q_l) - y = \Parikh_{w_1}[q_l]$
$\Leftrightarrow$
$0 \leq x = max(p_l) - \Parikh_{w_2}[p_l] = \Parikh_{w_2}[p_r] - min(p_r) \leq \bp_1(i)-i$ and
$0 \leq y = max(q_l) - \Parikh_{w_1}[q_l] = \Parikh_{w_1}[q_r] - min(q_r) \leq \bp_2(k)-k$.

Assume $0 \leq x = max(p_l) - \Parikh_{w_2[k..k+d-1]}[p_l] = \Parikh_{w_2[k..k+d-1]}[p_r] - min(p_r) \leq \bp_1(i)-i$ and $0 \leq y = max(q_l) - \Parikh_{w_1[i..i+d-1]}[q_l] = \Parikh_{w_1[i..i+d-1]}[q_r] - min(q_r) \leq \bp_2(k)-k$.
Then, $x = \Parikh_{w_1[i..i+d-1]}[p_l] - \Parikh_{w_2[k..k+d-1]}[p_l] = \Parikh_{w_2[k..k+d-1]}[p_r] - \Parikh_{w_2[k..k+d-1]}[q_r]$ 
and $y = \Parikh_{w_2[k..k+d-1]}[q_l] - \Parikh_{w_1[i..i+d-1]}[q_l] = \Parikh_{w_1[i..i+d-1]}[q_r] - \Parikh_{w_2[k..k+d-1]}[q_r]$.
That is, $\Parikh_{w_2[k..k+d-1]}[p_l] = \Parikh_{w_2[k+y..k+y+d-1]}[p_l]= \Parikh_{w_1[i..i+d-1]}$ $[p_l] - x = \Parikh_{w_1[i+x..i+x+d-1]}[p_l]$,
$\Parikh_{w_2[k..k+d-1]}[p_r] = \Parikh_{w_2[k+y..k+y+d-1]}[p_r] = \Parikh_{w_1[i..i+d-1]}[p_r] + x = \Parikh_{w_1[i+x..i+x+d-1]}[p_r]$,
$\Parikh_{w_1[i..i+d-1]}[q_l] = \Parikh_{w_1[i+x..i+x+d-1]}$ $[q_l] = \Parikh_{w_2[k..k+d-1]}[q_l] - y = \Parikh_{w_2[k+y..k+y+d-1]}[q_l]$,
and $\Parikh_{w_1[i..i+d-1]}[q_r] = \Parikh_{w_1[i+x..i+x+d-1]}[q_r] = \Parikh_{w_2[k..k+d-1]}[q_r] + y = \Parikh_{w_2[k+y..k+y+d-1]}[q_r]$.
Therefore, a pair $(w_1[i+x..i+x+d-1], w_2[k+y..k+y+d-1])$ is a common Abelian factor of $w_1$ and $w_2$.
\end{proof}

\begin{lemma}\label{lem:p-s}
Assume $c_{q_l} \neq c_{p_l} = c_{p_r} \neq c_{q_r}$ and $c_{q_l} \neq c_{q_r}$.
There is a common Abelian factor $(w_1[i+x..i+x+d-1], w_2[k+y..k+y+d-1])$ of length $d$ iff  $0 \leq x \leq \bp_1(i)-i$, $0 \leq y = max(q_l) - \Parikh_{w_1[i..i+d-1]}[q_l] = \Parikh_{w_1[i..i+d-1]}[q_r] - min(q_r) \leq \bp_2(k)-k$ and $\Parikh_{w_1[i..i+d-1]}[p_l] = \Parikh_{w_2[k..k+d-1]}[p_l]$.
\end{lemma}

\begin{proof}\label{pro:p-s}
During the slide,
the number of occurrences of $c_{p_l}$~$(=c_{p_r})$ in the window
on $w_1$ does not change.

Assume that there is a common Abelian factor $(w_1[i+x..i+x+d-1], w_2[k+y..k+y+d-1])$.
Clearly $0 \leq x \leq \bp_1(i)-i$ and $0 \leq y \leq \bp_2(k)-k$.
Then, it holds that
$\Parikh_{w_2[k+y..k+y+d-1]}[q_l] = max(q_l) - y = \Parikh_{w_1[i+x..i+x+d-1]}[q_l]$,
$\Parikh_{w_2[k+y..k+y+d-1]}[q_r] = min(q_r) + y = \Parikh_{w_1[i+x..i+x+d-1]}[q_r]$ and
$\Parikh_{w_2[k+y..k+y+d-1]}$ $[p_l] =  \Parikh_{w_1[i+x..i+x+d-1]}[p_l]$,
that is,
$0 \leq y = max(q_l) - \Parikh_{w_1[i..i+d-1]}[q_l] = \Parikh_{w_1[i..i+d-1]}[q_r] - min(q_r)$ and $\Parikh_{w_1[i..i+d-1]}[p_l] = \Parikh_{w_2[k..k+d-1]}[p_l]$.

Assume $y = max(q_l) - \Parikh_{w_1[i..i+d-1]}[q_l] = \Parikh_{w_1[i..i+d-1]}[q_r] - min(q_r)$ and $\Parikh_{w_1[i..i+d-1]}[p_l] = \Parikh_{w_2[k..k+d-1]}[p_l]$.
Then, $\Parikh_{w_2[k..k+d-1]}[q_l] - y = \Parikh_{w_1[i..i+d-1]}[q_l]$ and $\Parikh_{w_2[k..k+d-1]}[q_r]+y = \Parikh_{w_1[i..i+d-1]}[q_r]$, 
that is, $\Parikh_{w_2[k+y..k+y+d-1]}[q_l] = \Parikh_{w_1[i..i+d-1]}[q_l]$ and $\Parikh_{w_2[k+y..k+y+d-1]}[q_r] = \Parikh_{w_1[i..i+d-1]}[q_r]$.
Therefore, a pair $(w_1[i+x..i+x+d-1], w_2[k+y..k+y+d-1])$ is a common Abelian factor of length $d$ of $w_1$ and $w_2$.
\end{proof}

\begin{lemma}\label{lem:q-s}
Assume $c_{p_l} \neq c_{q_l} = c_{q_r} \neq c_{p_r}$ and $c_{p_l} \neq c_{p_r}$.
There is a common Abelian factor $(w_1[i+x..i+x+d-1], w_2[k+y..k+y+d-1])$ of length $d$ iff $0 \leq y \leq \bp_2(k)-k$ and $x = max(p_l) - \Parikh_{w_2[k..k+d-1]}[p_l] = \Parikh_{w_2[k..k+d-1]}[p_r] - min(p_r) \geq 0$.
\end{lemma}
Lemma~\ref{lem:q-s} can be proved by a similar argument to the proof of Lemma~\ref{lem:p-s}.

\begin{lemma}\label{lem:backslash}
Assume $c_{p_r} \neq c_{p_l} = c_{q_r} \neq c_{q_l}$ and $c_{p_r} \neq c_{q_l}$.
There is a common Abelian factor $(w_1[i+x..i+x+d-1], w_2[k+y..k+y+d-1])$ of length $d$ iff $0 \leq x = \Parikh_{w_2[k..k+d-1]}[p_r] - min(p_r) \leq \bp_1(i)-i$, $0 \leq y = max(q_l) - \Parikh_{w_1[i..i+d-1]}[q_l] \leq \bp_2(k)-k$ and $x+y = \Parikh_{w_1[i..i+d-1]}[p_l] - \Parikh_{w_2[k..k+d-1]}[q_r] = max(p_l) - min(q_r)$.
\end{lemma}

\begin{proof}\label{pro:backslash}
During the slides of the windows,
the number of occurrences of $c_{q_l}$ in the window on $w_1$
and that of $c_{p_r}$ in the window on $w_2$ do not change.

Assume there is a common Abelian factor $(w_1[i+x..i+x+d-1], w_2[k+y..k+y+d-1])$.
Then,
$\Parikh_{w_1[i..i+d-1]}[q_l] = \Parikh_{w_2[k+y..k+y+d-1]}[q_l] = max(q_l) - y$,
$\Parikh_{w_2[k..k+d-1]}[p_r] = \Parikh_{w_1[i+x..i+x+d-1]}[p_r] = min(p_r) + x$,
$\Parikh_{w_1[i+x..i+x+d-1]}[p_l] = min(p_l) + x = \Parikh_{w_2[k+y..k+y+d-1]}[q_r] = max(q_r) - y$, 
that is, $y = max(q_l) - \Parikh_{w_1[i..i+d-1]}[q_l]$,
$x = \Parikh_{w_2[k..k+d-1]}[p_r] - min(p_r)$ and 
$x + y = max(p_l) - min(q_r)$.

Assume $y = max(q_l) - \Parikh_{w_1[i..i+d-1]}[q_l]$,
$x = \Parikh_{w_2[k..k+d-1]}[p_r] - min(p_r)$ and 
$x + y = max(p_l) - min(q_r)$.
Then, $y = \Parikh_{w_2[k..k+d-1]}[q_l] - \Parikh_{w_1[i..i+d-1]}[q_l]$,
$x = \Parikh_{w_2[k..k+d-1]}[p_r] - \Parikh_{w_1[i..i+d-1]}[p_r]$ and
$x+y = \Parikh_{w_1[i..i+d-1]}[p_l] - \Parikh_{w_2[k..k+d-1]}[q_r]$,
that is, 
$\Parikh_{w_1[i..i+d-1]}[q_l] = \Parikh_{w_1[i+x..i+x+d-1]}[q_l] = \Parikh_{w_2[k..k+d-1]}[q_l]-y = \Parikh_{w_2[k+y..k+y+d-1]}[q_l]$,
$\Parikh_{w_2[k..k+d-1]}[p_r] = \Parikh_{w_2[k+y..k+y+d-1]}[p_r] = \Parikh_{w_1[i..i+d-1]}$ $[p_r]+x = \Parikh_{w_1[i+x..i+x+d-1]}[p_r]$ and
$\Parikh_{w_1[i..i+d-1]}[p_l]-x = \Parikh_{w_1[i+x..i+x+d-1]}[p_l] = \Parikh_{w_2[k..k+d-1]}[q_r]+y = \Parikh_{w_2[k+y..k+y+d-1]}[q_r]$.
Therefore, a pair $(w_1[i+x..i+x+d-1], w_2[k+y..k+y+d-1])$ is a common Abelian factor of length $d$ of $w_1$ and $w_2$.
\end{proof}

\begin{lemma}\label{lem:slash}
Assume $c_{p_l} \neq c_{q_l} = c_{p_r} \neq c_{q_r}$ and $c_{p_l} \neq c_{q_r}$.
There is a common Abelian factor $(w_1[i+x..i+x+d-1], w_2[k+y..k+y+d-1])$ of length $d$ iff $0 \leq x = max(p_l) - \Parikh_{w_2[k..k+d-1]}[p_l] \leq \bp_1(i)-i$, $0 \leq y = \Parikh_{w_1[i..i+d-1]}[q_r] - min(q_r) \leq \bp_2(k)-k$ and $x+y = \Parikh_{w_2[k..k+d-1]}[q_l] - \Parikh_{w_1[i..i+d-1]}[p_r] = max(q_l) - min(p_r)$.
\end{lemma}

Lemma~\ref{lem:slash} can be proved by a similar argument to the proof of Lemma~\ref{lem:backslash}.

\begin{theorem}
  Given two strings $w_1$ and $w_2$, we can compute
  an $O(m^2)$-size representation of all
  longest common Abelian factors of $w_1$ and $w_2$
  in $O(m^2n)$ time with $O(\sigma)$ working space,
  where $m$ and $n$ are the total size of the RLEs and the total length
  of $w_1$ and $w_2$, respectively.
\end{theorem}

\begin{proof}
The correctness follows from Lemmas~\ref{lem:all}--\ref{lem:slash}.

Let $m_1, m_2$ be the sizes of $\RLE(w_1)$ and $\RLE(w_2)$,
respectively.
Let $n_{\min} = \min\{|w_1|, |w_2|\}$.
For each fixed window size $d$,
the window for $w_1$ shifts over $w_1$ $O(m_1)$ times.
For each shift of the window for $w_1$,
the window for $w_2$ shifts over $w_2$ $O(m_2)$ times.
Thus, we have $O(m_1 \cdot m_2 \cdot n_{\min})$ total shifts.
Since all the conditions in Lemmas~\ref{lem:all}--\ref{lem:slash}
can be tested in $O(1)$ time each by simple arithmetic,
the total time complexity is $O(m_1 m_2 n_{\min} + n)$,
where the $n$ term denotes the cost to compute $\RLE(w_1)$ and $\RLE(w_2)$.
Thus, it is clearly bounded by $O(m^2n)$.
Next, we focus on the output size.
Let $l$ be the length of the longest common Abelian factors
of $w_1$ and $w_2$.
Using Lemmas~\ref{lem:II}--\ref{lem:slash},
for each pair of the shifts of the two windows
we can compute an $O(1)$-size representation of the longest
common Abelian factors found.
Since there are $O(m_1 \cdot m_2)$ shifts for window length $l$,
the output size is bounded by $O(m^2)$.
The working space is $O(\sigma)$,
since we only need to maintain two Parikh vectors
for the two sliding windows.
\end{proof}

\subsection{Example for Computin Longest Common Abelian facotors using RLEs}

We show an example of how our algorithm computes a common Abelian factor
of length $4$ for two input strings $w_1 = aaaaacbbbcc$ and $w_2 = cccaaccbbbb$.

\begin{figure}[!h]
 \begin{minipage}{0.45\hsize}
  \begin{center}
   \includegraphics[width=50mm]{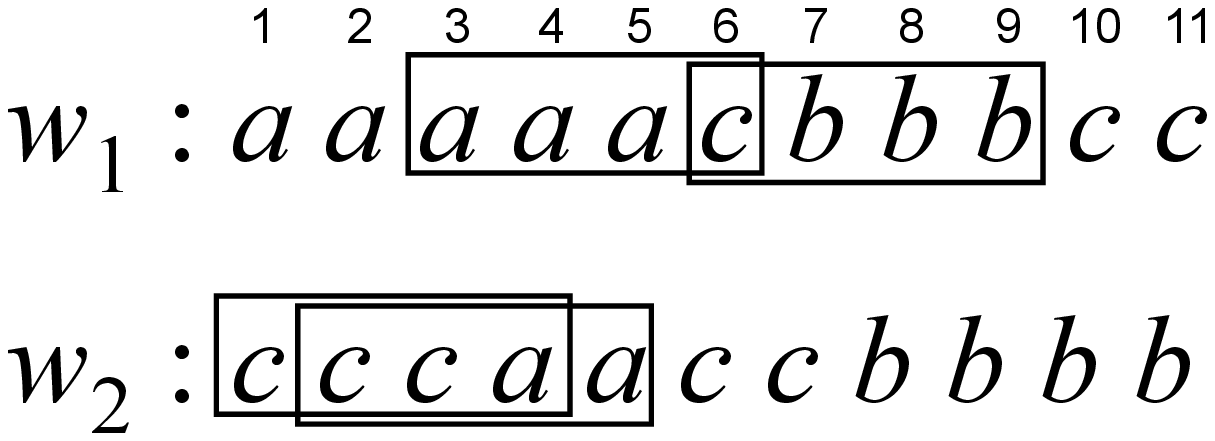}
  \end{center}
  \caption{Showing two sliding windows of length $d=4$,
    where $i=3$, $\bp_1(i)=6$, $k=1$, $\bp_2(k)=2$, $c_{p_l}=a$, $c_{p_r}=b$, $c_{q_l}=c$, $c_{q_r}=a$.}
  \label{fig:lcaf1}
 \end{minipage}
 \hfill
 \begin{minipage}{0.45\hsize}
  \begin{center}
   \includegraphics[width=50mm]{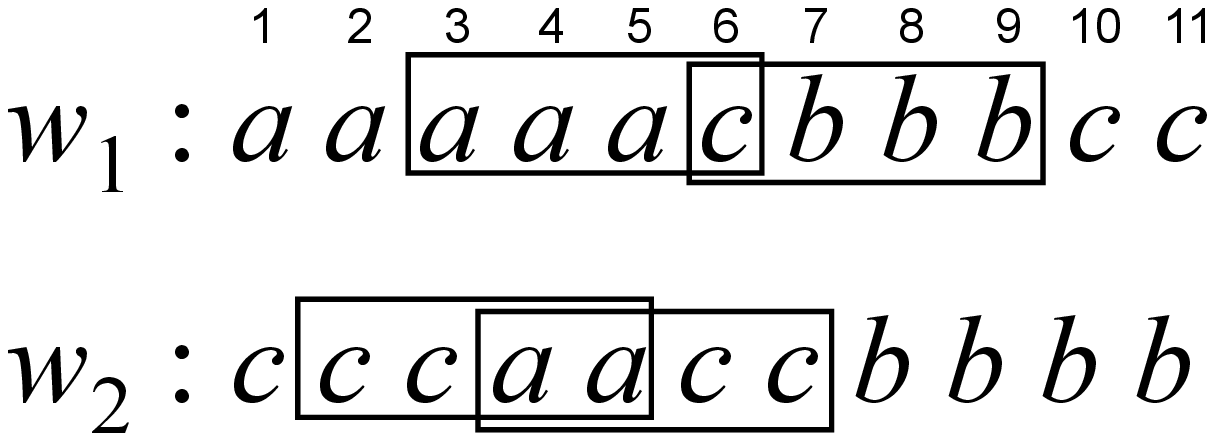}
  \end{center}
  \caption{Showing two sliding windows of length $d=4$,
    where $i=3$, $\bp_1(i)=6$, $k=2$, $\bp_2(k)=4$, $c_{p_l}=a$, $c_{p_r}=b$, $c_{q_l}=c$, $c_{q_r}=c$.}
  \label{fig:lcaf2}
 \end{minipage}
\end{figure}

Suppose that the window for $w_1$ is now aligned at
position $3$ of $w_1$ (namely $U_3 = w_1[3..6] = aaac$).
We then shift it to position $\bp_1(3) = 6$
(namely $U_6 = w_1[6..9] = cbbb$).
For this shift of the window on $w_1$,
we test $O(m_2)$ shifts of the window over the second string $w_2$,
as follows.

We begin with position $1$ of the other string $w_2$
(namely $V_1 = w_2[1..4] = ccca$),
and shift the window to position $\bp_2(1) = 2$.
See also Figure~\ref{fig:lcaf1}.
It follows from Lemma~\ref{lem:backslash} that 
there is no common Abelian factor during these slides.
We move on to the next step.

Next, the window for $w_2$ is shifted from position $2$
to position $\bp_2(2) = 4$ (namely, $V_4 = w_2[4..7] = aacc$).
See also Figure~\ref{fig:lcaf2}.
It follows from Lemma~\ref{lem:q-s} that
there is no common Abelian factor during the slides.
We move on to the next step.

\begin{figure}[!h]
 \begin{minipage}{0.45\hsize}
  \centering{
   \includegraphics[width=50mm]{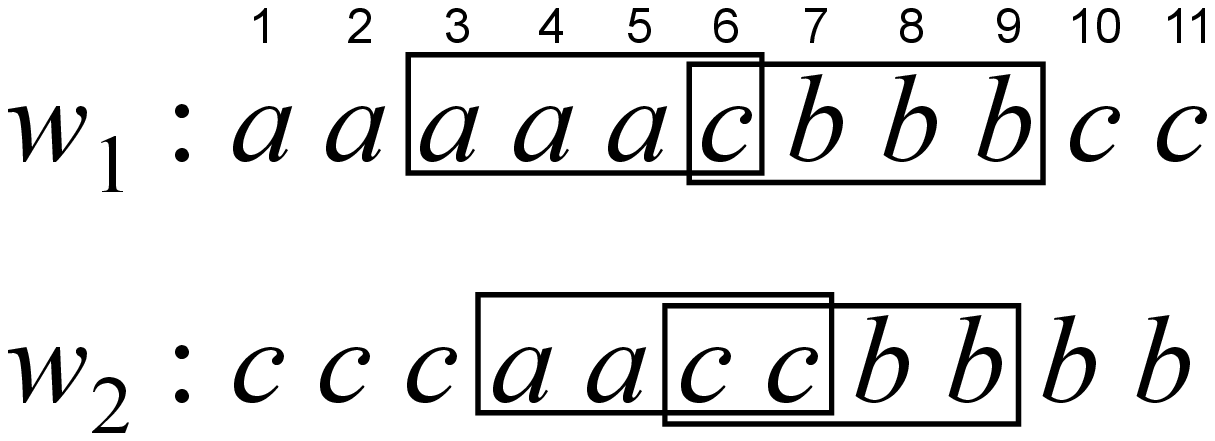}
   \caption{Showing two sliding windows of length $d=4$,
     where $i=3$, $\bp_1(i)=3$, $k=4$, $\bp_2(k)=6$, $c_{p_l}=a$, $c_{p_r}=b$, $c_{q_l}=a$, $c_{q_r}=b$.}
  }
  \label{fig:lcaf3}
 \end{minipage}
 \hfill
 \begin{minipage}{0.45\hsize}
  \centering{
   \includegraphics[width=50mm]{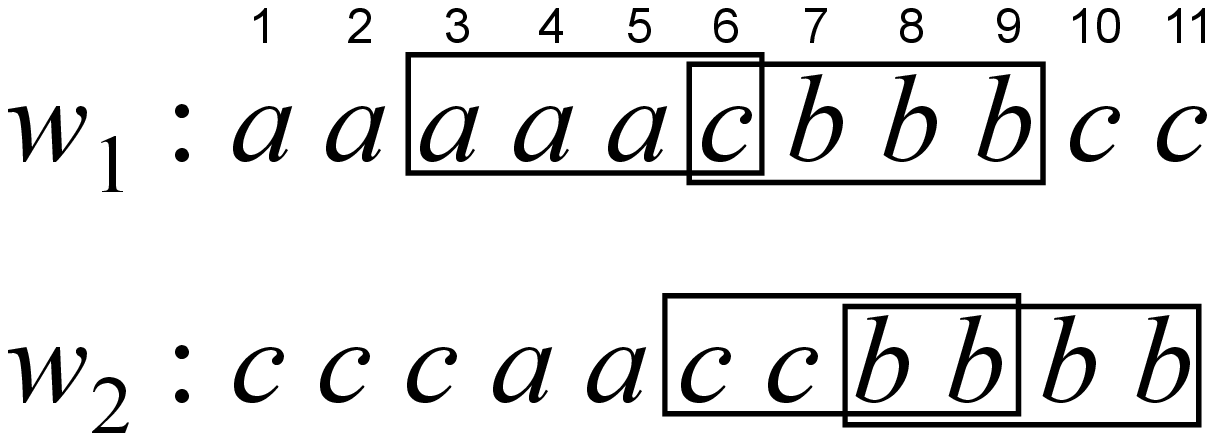}
   \caption{Showing two sliding windows of length $d=4$,
     where $i=3$, $\bp_1(i)=3$, $k=6$, $\bp_2(k)=3$, $c_{p_l}=a$, $c_{p_r}=b$, $c_{q_l}=c$, $c_{q_r}=b$.}
  }
  \label{fig:lcaf4}
 \end{minipage}
\end{figure}

Next, the window for $w_2$ is shifted from position $4$
to position $\bp_2(4) = 6$ (namely, $V_6 = w_2[6..9] = ccbb$).
See also Figure~\ref{fig:lcaf3}.
Since the numbers of occurrences of $c$ on $w_1$ and $w_2$ are different
and $c$ is not equal to $a$ or $b$,
there is no common Abelian factor during the slides.
We move on to the next step.

Next, the window for $w_2$ is shifted from position $6$
to position $\bp_2(6) = 8$.
See Figure~\ref{fig:lcaf4}.
It follows from Lemma~\ref{lem:right} that 
there is a common Abelian factor $(w_1[6..9], w_2[7..10])$
of length $d = 4$.

\bibliographystyle{abbrv}
\bibliography{ref}

\end{document}